\documentclass{article}
\usepackage{arxiv}
\usepackage[accsupp]{axessibility}

\pdfoutput=1

\usepackage[utf8]{inputenc}
\usepackage{inconsolata}

\usepackage[T1]{fontenc}

\usepackage{wrapfig}
\usepackage{lipsum}
\usepackage{listings}
\usepackage{graphicx}
\usepackage{xcolor}
\usepackage{float}
\usepackage{caption}
\usepackage{tikz}
\usepackage{amssymb,amsmath,amsthm}

\definecolor{lgreen} {RGB}{180,210,100}
\definecolor{dblue}  {RGB}{20,66,129}
\definecolor{ddblue} {RGB}{11,36,69}
\definecolor{lred}   {RGB}{220,0,0}
\definecolor{nred}   {RGB}{224,0,0}
\definecolor{norange}{RGB}{230,120,20}
\definecolor{nyellow}{RGB}{255,221,0}
\definecolor{ngreen} {RGB}{98,158,31}
\definecolor{dgreen} {RGB}{78,138,21}
\definecolor{nblue}  {RGB}{28,130,185}
\definecolor{jblue}  {RGB}{20,50,100}

\definecolor{GreenYellow}       {RGB}{217, 229, 6} 	    
\definecolor{Yellow}            {RGB}{254, 223, 0} 	    
\definecolor{Goldenrod}         {RGB}{249, 214, 22} 	
\definecolor{Dandelion}         {RGB}{253, 200, 47} 	
\definecolor{Apricot}           {RGB}{255, 170, 123} 	
\definecolor{Peach}             {RGB}{255, 127, 69} 	
\definecolor{Melon}             {RGB}{255, 129, 141} 	
\definecolor{YellowOrange}      {RGB}{240, 171, 0} 	    
\definecolor{Orange}            {RGB}{255, 88, 0} 	    
\definecolor{BurntOrange}       {RGB}{199, 98, 43} 	    
\definecolor{Bittersweet}       {RGB}{189, 79, 25} 	    
\definecolor{RedOrange}         {RGB}{222, 56, 49} 	    
\definecolor{Mahogany}          {RGB}{152, 50, 34} 	    
\definecolor{Maroon}            {RGB}{152, 30, 50} 	    
\definecolor{BrickRed}          {RGB}{170, 39, 47} 	    
\definecolor{Red}               {RGB}{255, 0, 0}        
\definecolor{BrilliantRed}      {RGB}{237, 41, 57} 	    
\definecolor{OrangeRed}         {RGB}{231, 58, 0} 	    
\definecolor{RubineRed}         {RGB}{202, 0, 93}       
\definecolor{WildStrawberry}    {RGB}{203, 0, 68} 	    
\definecolor{Salmon}            {RGB}{250, 147, 171} 	
\definecolor{CarnationPink}     {RGB}{226, 110, 178} 	
\definecolor{Magenta}           {RGB}{255, 0, 144} 	    
\definecolor{VioletRed}         {RGB}{215, 31, 133} 	
\definecolor{Rhodamine}         {RGB}{224, 17, 157} 	
\definecolor{Mulberry}          {RGB}{163, 26, 126} 	
\definecolor{RedViolet}         {RGB}{161, 0, 107} 	    
\definecolor{Fuchsia}           {RGB}{155, 24, 137} 	
\definecolor{Lavender}          {RGB}{240, 146, 205} 	
\definecolor{Thistle}           {RGB}{222, 129, 211} 	
\definecolor{Orchid}            {RGB}{201, 102, 205} 	
\definecolor{DarkOrchid}        {RGB}{153, 50, 204} 	
\definecolor{Purple}            {RGB}{182, 52, 187} 	
\definecolor{Plum}              {RGB}{79, 50, 76} 	    
\definecolor{Violet}            {RGB}{75, 8, 161} 	    
\definecolor{RoyalPurple}       {RGB}{82, 35, 152} 	    
\definecolor{BlueViolet}        {RGB}{33, 7, 106} 	    
\definecolor{Periwinkle}        {RGB}{136, 132, 213} 	
\definecolor{CadetBlue}	  	    {RGB}{95, 158, 160} 	
\definecolor{CornflowerBlue}  	{RGB}{99, 177, 229} 	
\definecolor{MidnightBlue}	  	{RGB}{0, 65, 101} 	    
\definecolor{NavyBlue}          {RGB}{0, 70, 173}       
\definecolor{RoyalBlue}         {RGB}{0, 35, 102}       
\definecolor{Blue}              {RGB}{0, 24, 168}       
\definecolor{Cerulean}          {RGB}{0, 122, 201}      
\definecolor{Cyan}              {RGB}{0, 159, 218}      
\definecolor{ProcessBlue}       {RGB}{0, 136, 206}      
\definecolor{SkyBlue}           {RGB}{91, 198, 232}     

\definecolor{Turquoise}         {RGB}{0, 255, 239} 	    

\definecolor{TealBlue}          {RGB}{0, 124, 146} 	    
\definecolor{Aquamarine}        {RGB}{0, 148, 179} 	    
\definecolor{BlueGreen}         {RGB}{0, 154, 166} 	    
\definecolor{Emerald}           {RGB}{80, 200, 120} 	
\definecolor{JungleGreen}       {RGB}{0, 115, 99} 	    
\definecolor{SeaGreen}          {RGB}{0, 176, 146} 	    
\definecolor{Green}             {RGB}{0, 173, 131} 	    
\definecolor{ForestGreen}       {RGB}{0, 105, 60} 	    
\definecolor{PineGreen}         {RGB}{0, 98, 101} 	    
\definecolor{LimeGreen}         {RGB}{50, 205, 50} 	    
\definecolor{YellowGreen}       {RGB}{146, 212, 0} 	    
\definecolor{SpringGreen}       {RGB}{201, 221, 3} 	    
\definecolor{OliveGreen}        {RGB}{135, 136, 0} 	    
\definecolor{RawSienna}         {RGB}{149, 82, 20} 	    
\definecolor{Sepia}             {RGB}{98, 60, 27} 	    
\definecolor{Brown}             {RGB}{134, 67, 30}      
\definecolor{Tan}               {RGB}{210, 180, 140}	
\definecolor{Gray}              {RGB}{139, 141, 142} 	

\definecolor{Black}		  	    {RGB}{30, 30, 30}       
\definecolor{White}		  	    {RGB}{255, 255, 255}    

\newcommand{\hlcolor}{Yellow!35}
\newcommand{\statecolor}{Black}
\newcommand{\hlcolorTwo}{BrickRed!35}
\makeatletter
\newenvironment{btHighlight}[1][]
{\begingroup\tikzset{bt@Highlight@par/.style={#1}}\begin{lrbox}{\@tempboxa}}
{\end{lrbox}\bt@HL@box[bt@Highlight@par]{\@tempboxa}\endgroup}

\newcommand\btHL[1][]{%
  \begin{btHighlight}[#1]\bgroup\aftergroup\bt@HL@endenv%
}
\def\bt@HL@endenv{%
  \end{btHighlight}%
  \egroup
}
\newcommand{\bt@HL@box}[2][]{%
  \tikz[#1]{%
    \pgfpathrectangle{\pgfpoint{1pt}{0pt}}{\pgfpoint{\wd #2}{\ht #2}}%
    \pgfusepath{use as bounding box}%
    \node[anchor=base west, fill=\hlcolor,outer sep=0pt,inner xsep=1pt, inner ysep=0pt, rounded corners=2pt, minimum height=\ht\strutbox+2pt,#1]{\raisebox{1pt}{\strut}\strut\usebox{#2}};
  }%
}

\newenvironment{btHighlightTwo}[1][]
{\begingroup\tikzset{bt@HighlightTwo@par/.style={#1}}\begin{lrbox}{\@tempboxa}}
{\end{lrbox}\bt@HLTwo@box[bt@HighlightTwo@par]{\@tempboxa}\endgroup}

\newcommand\btHLTwo[1][]{%
  \begin{btHighlightTwo}[#1]\bgroup\aftergroup\bt@HLTwo@endenv%
}
\def\bt@HLTwo@endenv{%
  \end{btHighlightTwo}%
  \egroup
}
\newcommand{\bt@HLTwo@box}[2][]{%
  \tikz[#1]{%
    \pgfpathrectangle{\pgfpoint{1pt}{0pt}}{\pgfpoint{\wd #2}{\ht #2}}%
    \pgfusepath{use as bounding box}%
    \node[anchor=base west, fill=\hlcolorTwo,outer sep=0pt,inner xsep=1pt, inner ysep=0pt, rounded corners=2pt, minimum height=\ht\strutbox+2pt,#1]{\raisebox{1pt}{\strut}\strut\usebox{#2}};
  }%
}

\usepackage{etoolbox}

\newcommand{\indentrule}{\color{code_indent}\vrule\hspace{2pt}}
\definecolor{code_indent}{HTML}{CCCCCC}

\newenvironment{figureAsListingWide}
    {
    \addtocounter{figure}{-1}
    \refstepcounter{lstlisting}
     
    \begin{figure*}[!htbp]
        
        \centering
    }
    { 
        \end{figure*} 
    }

\renewcommand{\paragraph}[1]{\vspace{0.3em}\noindent\textbf{#1}}

\newtheorem{theorem}{Theorem}
\newtheorem{lemma}{Lemma}

\usepackage{hyperref}

\title{Fast and Scalable Channels in Kotlin Coroutines}

\author{
  Nikita Koval \\
  JetBrains\\
  \texttt{nikita.koval@jetbrains.com}
  \And
  Dan Alistarh\\
  IST Austria\\
  \texttt{dan.alistarh@ist.ac.at} \\
    \And
  Roman Elizarov \\
  JetBrains\\
  \texttt{roman.elizarov@jetbrains.com}
}

\begin{document}

\maketitle

\begin{abstract}
Asynchronous programming has gained significant popularity over the last decade: support for this programming pattern is available in many popular languages via libraries and native language implementations, typically in the form of coroutines or the \texttt{async/await} construct. 
Instead of programming via shared memory, this concept assumes implicit synchronization through message passing. 
The key data structure enabling such communication is the \emph{rendezvous channel}. 
Roughly, a rendezvous channel is a blocking queue of size zero, so both \texttt{send(e)} and \texttt{receive()} operations wait for each other, performing a rendezvous when they meet. 
To optimize the message passing pattern, channels are usually equipped with a fixed-size buffer, so \texttt{send}s do not suspend and put elements into the buffer until its capacity is exceeded. This primitive is known as a \emph{buffered channel}. 

This paper presents a fast and scalable algorithm for both rendezvous and buffered channels. Similarly to modern queues, our solution is based on an infinite array with two positional counters for \texttt{send(e)} and \texttt{receive()} operations, leveraging the unconditional \texttt{Fetch-And-Add} instruction to update them. Yet, the algorithm requires non-trivial modifications of this classic pattern, in order to support the full channel semantics, such as buffering and cancellation of waiting requests. 
We compare the performance of our solution to that of the Kotlin implementation, as well as against other academic proposals, showing up to \texttt{9.8}$\times$ speedup.
To showcase its expressiveness and performance, we also integrated the proposed algorithm into the standard Kotlin Coroutines library, replacing the previous channel implementations.
\end{abstract}

\section{Introduction}\label{section:intro}
Asynchronous programming is a classic construct that allows to detach the execution of a code unit from the primary application. It is extremely useful when executing code with very high variance in its latency, such as I/O or networking. 
\emph{Coroutines}~\cite{kahn1976coroutines} are a general and efficient way of supporting asynchronous programming. They may suspend and resume code execution, thus allowing for non-preemptive multitasking. Programming with coroutines is similar to standard thread-based programming, with the key difference that coroutines are lightweight and reuse native threads. 
Thus, a programmer can create a coroutine for each small task and combine these tasks using structured concurrency. In contrast with threads, which are preemptively multi-tasked by the scheduler and require additional synchronization, coroutines are \emph{cooperatively} multi-tasked {---} they explicitly yield control of the execution flow and do not require additional synchronization. 
Many modern programming languages, such as Kotlin, Go, and C++, are flexible and efficient enough to offer native support for coroutines.  

The key synchronization primitive behind coroutines is the \emph{rendezvous channel}, also known as a \textit{synchronous queue}. 
At its core, a rendezvous channel is a blocking bounded queue of zero capacity.  
It supports \texttt{send(e)} and \texttt{receive()} requests: \texttt{send(e)} checks whether the queue contains any receivers and either removes the first one (i.e. performs a ``rendezvous'' with it) or adds itself to the queue as a waiting sender; the \texttt{receive()} operation works symmetrically. Note that these operations are blocking by design. 
The rendezvous channel serves as the main synchronization primitive and therefore is key to asynchronous programming, providing the basis for other message-passing protocols, such as communicating sequential processes (CSP)~\cite{hoare1978communicating} and actors~\cite{agha1985actors}. 

To provide better insight into the channel semantics, consider an example. The figure below shows the classic producer\nobreakdash-consumer scenario with multiple producers sending tasks to consumers through a shared rendezvous channel. 
Here, the channel is similar to a queue: tasks are added to one end and received from the other.
Let the first consumer begin by invoking \texttt{receive()} to retrieve a task: as the channel is empty, it suspends, waiting for a producer. When the producer comes, it makes a rendezvous with the suspended consumer, transferring the task directly to it. However, if one more producer comes in while no consumer is waiting on the channel, the \texttt{send(e)} invocation suspends. The next \texttt{receive()} call by a consumer makes a rendezvous with the suspended producer, taking the task to process. 

\begin{figure}[H]
    \vspace{-1em}
    \centering
    \includegraphics[width=0.5\linewidth]{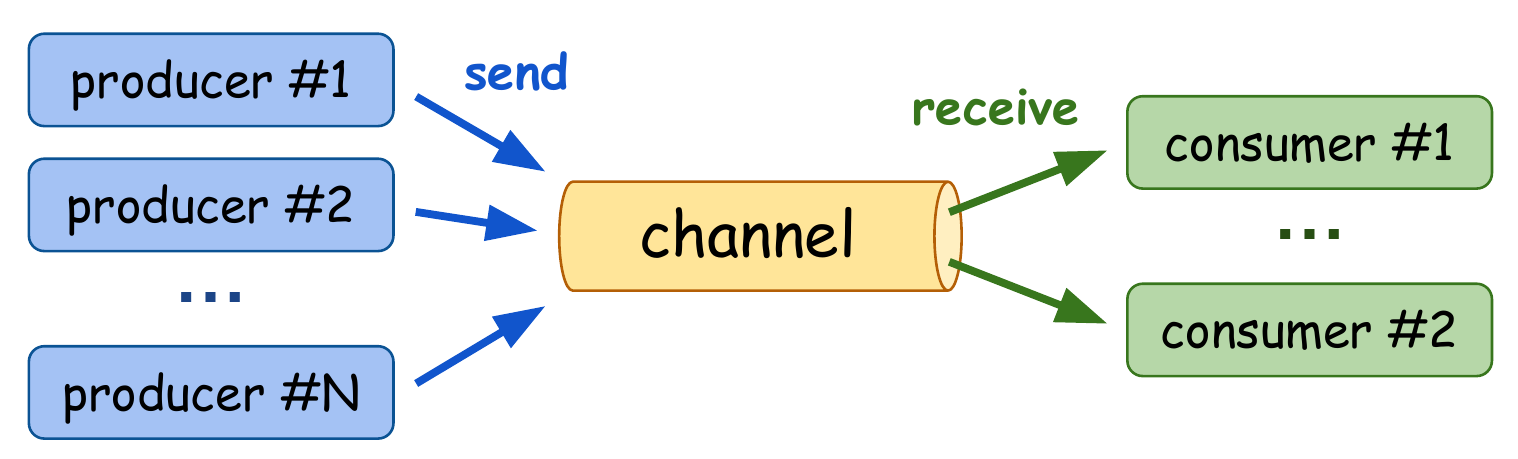}
    \vspace{-1em}
\end{figure}

When using channels for passing messages, as in the example above, programmers usually employ a more efficient \emph{buffered channel} primitive.
Roughly, it is a rendezvous channel equipped with a fixed-sized buffer: producers place elements in this buffer until its capacity is exceeded. Once the channel is full, new \texttt{send(e)} calls suspend until space in the buffer becomes available. The semantics are similar to a blocking queue of \emph{bounded capacity}.
Having limitations on the buffer size prevents a situation where producers are faster than consumers, growing the buffer indefinitely. An efficient and scalable implementation of a buffered channel is crucial for high-performance message passing.

Channels are useful, and a non-trivial amount of research has been invested in providing fast implementations. 
Java provides an efficient lock-free \texttt{SynchronousQueue} data structure with the rendezvous channel semantics~\cite{scherer2006scalable}; essentially, it stores suspended requests in a Michael-Scott queue~\cite{michael1996simple}.
Kotlin provides an implementation based on an optimized version of a lock-free doubly-linked list~\cite{kotlincoroutines} with descriptors~\cite{harris2002practical} to ensure atomicity.
Most other languages, such~as~Go and Rust, leverage a coarse-grained locking design~\cite{golang, rust}.
Recent works by Izraelevitz and Scott\footnote{The synchronous queue algorithm by Izraelevitz and Scott~\cite{SPDQ} breaks the rendezvous channel semantics in terms of the conditions under which operations suspend; we present such an example interleaving in Appendix~\ref{appendix:mpdq}.}~\cite{SPDQ} and by Koval et al.~\cite{koval2019channels} proposed novel implementations of the rendezvous channel, leveraging innovations in the design of non-blocking queues~\cite{LCRQ,YM16}. These works do not support buffering. 

There are two key challenges when building efficient channels. 
The first is to efficiently maintain a queue of suspended operations, and to implement the procedure which allows each operation to check whether there is a waiting operation of the opposite type, with which it can rendezvous; if a rendezvous is not possible, the operation should reserve a new queue cell in which it can suspend.
When extending the rendezvous semantics to \emph{buffered} channels, interruption support becomes the main issue  {---} in practice, suspended requests can be canceled, becoming unavailable for rendezvous. Supporting both interruptions and buffering semantics is highly non-trivial: all practical solutions, such as the official buffered channel implementations in Kotlin~\cite{kotlincoroutines} and Go~\cite{golang}, still use a naive \emph{coarse-grained locking} design for synchronization.

\paragraph{Our contribution.}
This paper presents a new scalable and efficient approach to implementing rendezvous and buffered channels. 
We start from the idea of building a logically-infinite array~\cite{LCRQ, SPDQ}, which that stores suspended requests, together with two atomic counters which keep track of the total number of \texttt{send(e)} and \texttt{receive()} invocations ever performed. The counters should be incremented at the beginning of each operation {---} after that, the algorithm is able to decide whether the current operation should suspend or make a rendezvous with an operation of the opposite kind. When an operation increments the counter, it also ``reserves'' the corresponding cell in the infinite array. The remaining synchronization is performed in this cell, which can only be processed by one sender and one receiver.

So far, our design follows patterns set forth by previous work. Our main contribution is in supporting \emph{buffered channel} semantics, which introduce non-trivial difficulties. Roughly, at the logical level, we need to add an extra counter which should indicate the end of the logical buffer in the array, so senders can check whether the reserved cell is in the buffer or not. 
At the same time, receivers should move this counter forward, resuming suspended senders if required.
Although this logic appears simple, the resulting synchronization patterns for buffered channels have to be implemented extremely carefully to maintain the correctness of \texttt{send(e)} and \texttt{receive()} operations, especially when suspended operations can be interrupted.


We implemented our channel design in Kotlin and integrated it to implement the communication mechanisms underlying the Kotlin Coroutines library. 
Compared to the native Kotlin solution and previous synchronous queue implementations~\cite{koval2019channels,scherer2006scalable}, our algorithm is more scalable, and outperforms prior proposals by up to \texttt{9.8}$\times$ in terms of throughput. Notably, our approach is portable, so it should extend to other languages, such as Go or Rust.

\section{Environment}\label{section:preliminaries}
\paragraph{Coroutines management.}
The algorithm we present in the paper can manipulate both threads and coroutines. As channels are typically used for asynchronous programming,~we focus on coroutines as a use-case.
Since both \texttt{send(e)} and \texttt{receive()} operations are blocking, there should be a mechanism to suspend and schedule coroutines. The API we use in the paper is presented in Listing~\ref{listing:threads_management} and can be easily adapted to most programming languages, such as Kotlin, Go, or Java. 

\begin{lstlisting}[label={listing:threads_management}, caption={Coroutines management API}, numbers=none] 
interface Coroutine {
 fun tryUnpark(): Boolean^\label{line:tryUnpark}^
 fun interrupt()^\label{line:interrupt}^
 fun park(onInterrupt: lambda () -> Unit)^\label{line:park}^
}
fun curCor(): Coroutine^\label{line:curCor}^
\end{lstlisting}

Like in operating systems, we call \texttt{park(..)} for suspending the running coroutine that can be obtained by \texttt{curCor()} call. When parking a coroutine, the corresponding native thread does not stop but schedules another coroutine, making the coroutine suspension mechanism relatively cheap.

While being parked, we assume that the coroutine can be cancelled via \texttt{interrupt()} call, becoming unable to resume. Note that coroutine interruptions, unlike thread interruptions, happen relatively often, and, therefore, should be efficient. 
After a suspended coroutine is interrupted, the \texttt{onInterrupt} action set in \texttt{park(..)} is executed (line~\ref{line:park}). If a coroutine is interrupted in an active state, the interruption takes effect with the following \texttt{park(..)} invocation.

In order to resume a coroutine, the \texttt{tryUnpark()} function should be called. This function returns \texttt{true} if the resumption is successful, and \texttt{false} if the coroutine has already been interrupted. Note, that \texttt{tryUnpark()} can be called before \texttt{park(..)} {---} in this case the following \texttt{park(..)} invocation completes without suspension.


\paragraph{Memory model and memory reclamation.}
For simplicity, we assume the sequentially-consistent memory model. Additionally to plain reads and writes, we use atomic \texttt{Compare-} \texttt{And-Swap\:(CAS)} and \texttt{Fetch-And-Add\:(FAA)} instructions. 
We also assume that the runtime environment supports garbage collection (GC).
Reclamation techniques such as hazard pointers \cite{michael2004hazard} can be used in environments without GC.


\section{The Buffered Channel Algorithm}\label{section:algo}
The high-level structure of our algorithms will be based on the common abstraction of an \emph{infinite array} data structure. 
This infinite array stores elements and coroutines, with separate counters for all the \texttt{send(e)} and \texttt{receive()} invocations that were ever performed. 
Additionally, we have a counter which indicates the end of the logical buffer in the infinite array. 
Both \texttt{send(e)} and \texttt{receive()} operations start by incrementing their counters, reserving the array cell referenced by the counter. For rendezvous channels, the operation either suspends, storing its coroutine in the cell, or makes a rendezvous with a request of the opposite type. Each cell can be processed by exactly one \texttt{send(e)} and one \texttt{receive()}. As for buffered channels, \texttt{send(s)} adds elements without suspension if the cell is in the logical buffer, while \texttt{receive()} is responsible for updating the end of the buffer. 

\setlength{\columnsep}{1.1em}
\begin{wrapfigure}[6]{r}{0.25\linewidth}
    \centering
    \vspace{-1em}
    \includegraphics[width=\linewidth]{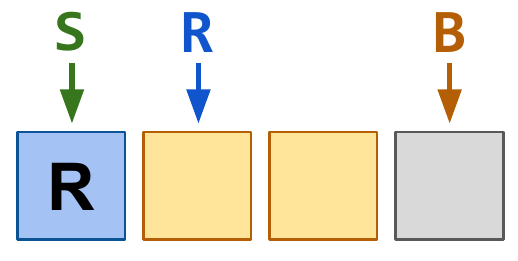}
\end{wrapfigure}

The embedded figure on the right shows a toy example of such a channel with a buffer of capacity $2$, and one \texttt{receive()} operation waiting for an element.
The following \texttt{send(e)} invocation increments the counter \texttt{S} and makes a rendezvous with the suspended receiver stored in the cell (marked with <<\texttt{R}>> in the figure). The next two \texttt{send(e)}-s proceed without suspending as the cells are in the logical buffer (buffer cells are highlighted in yellow). However, the fourth \texttt{send(e)} operation will have to suspend, as the buffer is already full.

Listing~\ref{listing:channel_structure} presents the buffered channel structure discussed above. 
The \texttt{S}, \texttt{R}, and \texttt{B} 64-bit counters store the total numbers of \texttt{send(e)} and \texttt{receive()} invocations and the end of the logical buffer. We assume that \texttt{send(e)} does not suspend when \texttt{s\:<\:B} (the buffer is not full) and makes a rendezvous with a waiting \texttt{receive()} when \texttt{s\:<\:R} (the location \texttt{A[s]} stores the receiver). Similarly, when \texttt{r\:<\:S}, \texttt{receive()} either performs a rendezvous with a waiting \texttt{send(e)} or retrieves an element from the buffer, suspending otherwise.
The infinite array \texttt{A} stores the cells with the buffered elements and suspended requests. Each cell is represented by a \texttt{Waiter} structure: the \texttt{state} field stores the suspended coroutine waiting for a rendezvous, while \texttt{elem} {---} the sending element.

\begin{lstlisting}[label={listing:channel_structure}, caption={The channel structure.}] 
class Channel<E> {
 ^\label{line:struct_sr}^var S, R, B: Long
 ^\label{line:struct_a}^var A: InfiniteArray<Waiter>
 fun send(element: E) { ... }
 fun receive(): E { ... }
}
struct Waiter<E>(state: Any?, elem: E?)
\end{lstlisting}

\subsection{The Rendezvous Channel}
We begin with an algorithm which supports only the rendezvous channel semantics; we will later extend this design to \emph{buffered} channels.
Both \texttt{send(e)} and \texttt{receive()} start by incrementing the corresponding counter via  \texttt{Fetch-And-Add}, thus, reserving the cell {---} at most one sender and one receiver can process each cell. Next, the operations read the counter of the opposite operation type to decide whether to make a rendezvous or suspend; in this part, we ignore the counter \texttt{B} and manipulate only \texttt{S} and~\texttt{R}. 
When the opposite counter ``covers'' the reserved cell (is greater than the reserved cell index), the operation makes a rendezvous. Otherwise, it installs the current coroutine to the \texttt{state} field and suspends.


\paragraph{Cell life-cycle.} 
Each \texttt{send-receive} pair working on the same cell synchronizes on the \texttt{state} field, whose state machine is presented in Figure~\ref{figure:rendezvous_cell}. To transfer the element, \texttt{send(e)} places it in the \texttt{elem} field before synchronization, so \texttt{receive()} can safely retrieve the element after that, following the \emph{safe publication} pattern.

\begin{figure}[b]
    \centering
    \includegraphics[width=0.7\linewidth]{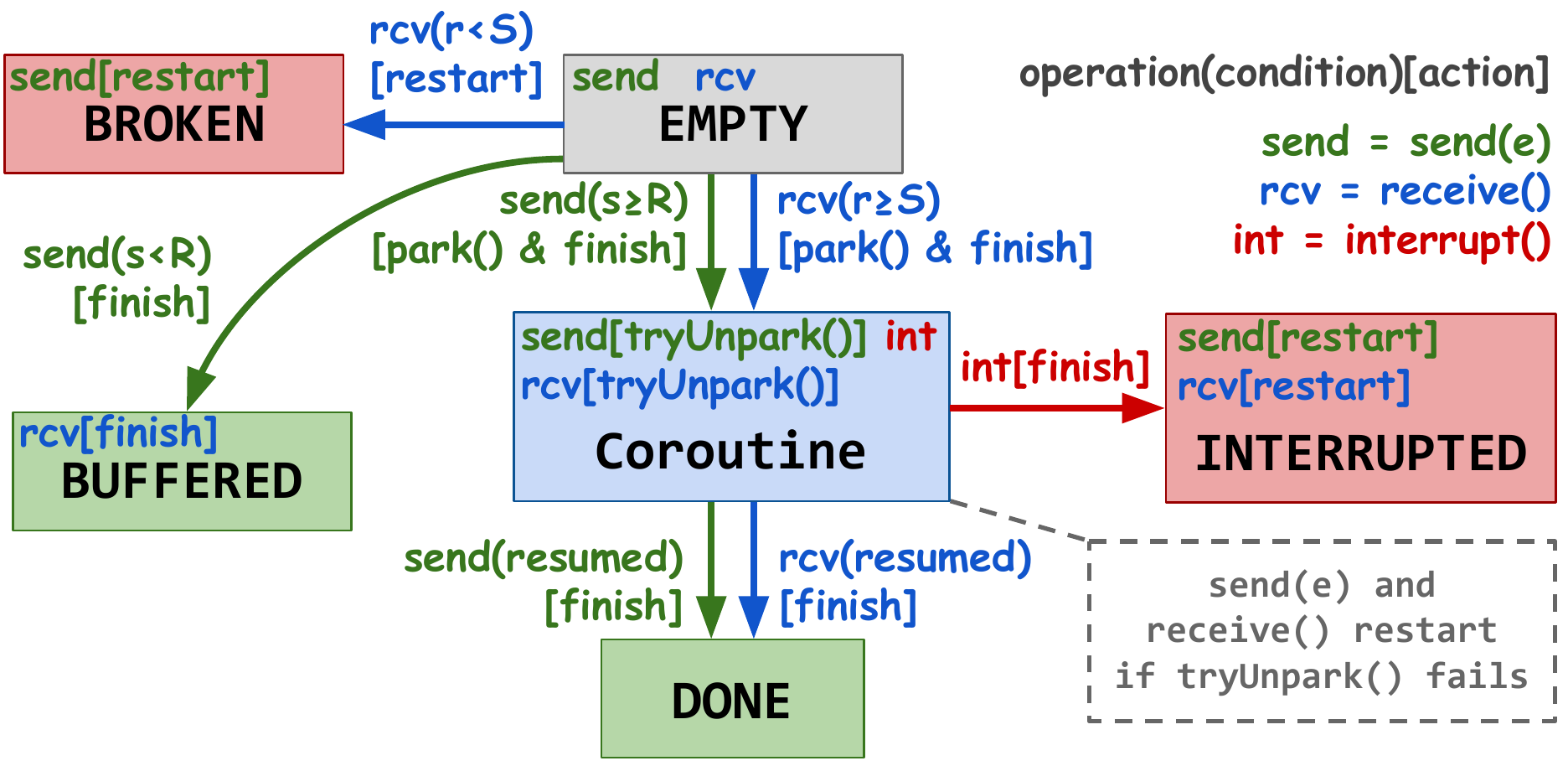}
    \captionof{figure}{
    {Cell life-cycle for the rendezvous channel.} 
    {Each state lists a set of operations that can discover it.
    If an operation should perform some action (e.g., finish or restart) when a state is observed, this action is specified in the square brackets. 
    Otherwise, the operation performs a transition to another state. When multiple transitions are possible, the condition for each is specified in circle brackets. Like the specification in states, the operation can perform the action specified in the square brackets when the transition succeeds. 
    The parameters \texttt{s} and \texttt{r} (in lower case) refer to the working cell indices for \texttt{send(e)} and \texttt{receive()}}. 
    }
    \label{figure:rendezvous_cell}
\end{figure}

Initially, each cell is in \texttt{EMPTY} state (denoted as \texttt{null} in the code). Usually, an operation to be suspended changes it to \texttt{Coroutine}. After that, an  operation of the opposite type comes to the cell, and makes a rendezvous by resuming the stored coroutine and updating the cell state to \texttt{DONE}.

\begin{figureAsListingWide}
\vspace{-0.3em}
\begin{minipage}[t]{0.49\textwidth}
\begin{lstlisting}[basicstyle=\scriptsize\selectfont\ttfamily]
fun send(element: E) = while(true) { 
 s := FAA(&S, +1) // reserve a cell ^\label{line:s:inc}^
 A[s].elem = element^\label{line:store_elem}^ 
 if updCellSend(s): return ^\label{line:s:updCell}^
}
// Returns `false` if this ^\color{Mahogany}`send(e)`^ should restart
fun updCellSend(s: Int): Bool = while(true) { ^\label{line:s:inf_loop}^
 state := A[s].state // read the current state ^\label{line:s:read_state}^
 r := R // read the receiver's counter ^\label{line:s:read_r}^
 when {
 ^\indentrule^// Empty and no receiver is coming ^\color{Mahogany}=>^ suspend
 ^\indentrule^state == null && s >= r: 
 ^\indentrule^  cor := curCor() // get the current coroutine^\label{line:s:curCor}^
 ^\indentrule^  if CAS(&A[s].state, null, cor): ^\label{line:s:store_cor}^
 ^\indentrule^  ^\indentrule^  cor.park( // wait for a ^\label{line:s:park0}^rendezvous
 ^\indentrule^  ^\indentrule^   onInterrupt = {A[s] = (INTERRUPTED, null)}) ^\label{line:s:handler}^ ^\label{line:s:park1}^
 ^\indentrule^  ^\indentrule^  return true ^\label{line:s:resumed}^
 ^\indentrule^// Waiting receiver ^\color{Mahogany}=>^ try to resume it
 ^\indentrule^state is Coroutine: 
 ^\indentrule^  if state.tryUnpark(): ^\label{line:s:tryUnpark}^
 ^\indentrule^  ^\indentrule^  A[s].state = DONE; return true ^\label{line:s:tryUnparkSuccess}^
 ^\indentrule^  else: // interrupted, clean the cell and fail
 ^\indentrule^  ^\indentrule^  A[s].elem = null; return false ^\label{line:s:tryUnparkFail}^
 ^\indentrule^// Empty but a receiver is coming ^\color{Mahogany}=>^ elimination
 ^\indentrule^state == null && s < r:
 ^\indentrule^  ^^if ^^@CAS(&A[s].state, null, BUFFERED)@: return true ^\label{line:s:markBuffered}^
 ^\indentrule^// Interrupted receiver or poisoned ^\color{Mahogany}=>^ fail
 ^\indentrule^state == INTERRUPTED || ^^#state == BROKEN#:  ^\label{line:s:intOrBroken0}^
 ^\indentrule^  A[s].elem = null // clean to avoid memory leaks 
 ^\indentrule^  return false ^\label{line:s:intOrBroken1}^
 }
}
\end{lstlisting}
\end{minipage}
\hfill
\begin{minipage}[t]{0.47\textwidth}
\begin{lstlisting}[firstnumber=33,basicstyle=\scriptsize\selectfont\ttfamily]
fun receive(): E = while(true) {
 r := FAA(&R, +1) // reserve a cell ^\label{line:rcvRend:inc_r}^
 if updCellRcv(r): 
 ^\indentrule^  e := A[r].elem; A[r].elem = null; return e
}
// Returns `false` if this ^\color{Mahogany}`receive()`^ should restart
fun updCellRcv(r: Int): Bool = while(true) {
 state := A[r].state // read the current state
 s := S // read the sender's counter
 when {
 ^\indentrule^// Empty and no sender is coming ^\color{Mahogany}=>^ suspend
 ^\indentrule^state == null && r >= s: 
 ^\indentrule^  cor := curCor() // get the current coroutine
 ^\indentrule^  if CAS(&A[r].state, null, cor):
 ^\indentrule^  ^\indentrule^  cor.park( // wait for a rendezvous
 ^\indentrule^  ^\indentrule^   onInterrupt = {A[r].state = INTERRUPTED})
 ^\indentrule^  ^\indentrule^  return true
 ^\indentrule^// Waiting sender ^\color{Mahogany}=>^ try to resume it
 ^\indentrule^state is Coroutine:
 ^\indentrule^  if state.tryUnpark():
 ^\indentrule^  ^\indentrule^  A[r].state = DONE; return true
 ^\indentrule^  else: // interrupted
 ^\indentrule^  ^\indentrule^  return false
 ^\indentrule^// Empty but a sender is coming ^\color{Mahogany}=>^ poison
 ^\indentrule^state == null && r < s:
 ^\indentrule^  if ^^#CAS(&A[r].state, null, BROKEN)#: return false
 ^\indentrule^// An elimination has happened ^\color{Mahogany}=>^ finish
 ^\indentrule^^^@state == BUFFERED@: return true
 ^\indentrule^// Interrupted sender ^\color{Mahogany}=>^ fail
 ^\indentrule^state == INTERRUPTED: return false
 }
}
\end{lstlisting}
\end{minipage}
\caption{The rendezvous channel algorithm. The \texttt{updCellSend(..)} and \texttt{updCellReceive(..)} functions update the cell state according to the diagram in Figure~\ref{figure:rendezvous_cell}.
The parts related to elimination and cell poisoning are highlighted with yellow~and~red.
}
\label{listing:send_highlevel}
\end{figureAsListingWide}

However, an operation that should resume the waiting request can come to \texttt{EMPTY} state if the operation to be suspended has not stored its coroutine yet. We use different approaches for \texttt{send(e)} and \texttt{receive()} to handle the race. The \texttt{send(e)} operation does an \emph{elimination} by changing the cell state to \texttt{BUFFERED}; it knows that a receiver is already coming, so there is no reason to wait for it. 
After that, a receiver comes and finds out that the elimination
has happened, taking the element and completing without suspension.

Unfortunately, we cannot use a similar logic for \texttt{receive()}, as it should not only resume the opposite operation but also retrieve the element. Instead, it ``poisons'' the cell by moving it to the \texttt{BROKEN} state, so both \texttt{send(e)} and \texttt{receive()} that work with the cell skip it and restart. This solution is reminiscent of the LCRQ queue~\cite{morrison2013fast}. 

Another reason to restart the operation is finding out that the coroutine of the opposite request is already interrupted; in this case, either the corresponding \texttt{tryUnpark()} fails, or the cell is in the \texttt{INTERRUPTED} state. When a coroutine is interrupted, it should change the cell state to \texttt{INTERRUPTED} to avoid memory leaks {---} this logic~can be provided via the \texttt{onInterrupt} lambda parameter to \texttt{park(..)}.

\paragraph{Algorithm for \texttt{send(e)}.}
Listing~\ref{listing:send_highlevel} presents the pseudocode for both \texttt{send(e)} and \texttt{receive()}. 
To send an element, the operation first increments its \texttt{S} counter and gets the value right before the increment via \texttt{Fetch-And-Add} (line~\ref{line:s:inc}). Next, it places the element into the reserved cell \texttt{A[s]} (line~\ref{line:store_elem}). 
Finally, it updates the cell state according to the diagram on Figure~\ref{figure:rendezvous_cell} by invoking \texttt{updCellSend(..)} function (line~\ref{line:s:updCell}).

To be consistent with the life-cycle diagram, the implementation of \texttt{updCellSend(..)} is written in a state-based manner. The logic is wrapped in an infinite loop, obtaining~the current cell state and \texttt{R} value at the beginning (lines~\ref{line:s:read_state}\nobreakdash--\ref{line:s:read_r}).

When the cell state is \texttt{EMPTY} and no receiver arrives (\texttt{s\:$\geq$\:r}), the operation decides to suspend. Thus, it obtains the current coroutine reference (line~\ref{line:s:curCor}), replaces \texttt{EMPTY} with it via
atomic \texttt{CAS} instruction (line~\ref{line:s:store_cor}), and suspends via \texttt{cor.park(..)} invocation (lines~\ref{line:s:park0}--\ref{line:s:park1}). If the \texttt{CAS} that installs the coroutine fails, the operation restarts. After the operation is resumed by receiver, it successfully completes (line~\ref{line:s:resumed}). It is also possible for \texttt{send(e)} to be interrupted while being parked. We pass a special interrupt handler to the \texttt{cor.park(..)} call that moves the cell state to \texttt{INTERRUPTED} and cleans up the element field (line~\ref{line:s:handler}).

When the cell stores a \texttt{Coroutine} instance, we know that this is a waiting \texttt{receive()} operation, with which a rendezvous should happen. Thus, \texttt{send(e)} tries to resume it, invoking \texttt{tryUnpark()} (line~\ref{line:s:tryUnpark}). On success, the state moves to \texttt{DONE}, and the operation completes (line~\ref{line:s:tryUnparkSuccess}). Otherwise, if the receiver is already interrupted, the operation cleans \texttt{A[s].elem} to avoid a memory leak and restarts (line~\ref{line:s:tryUnparkFail}).

It is also possible that the cell is already ``covered'' by an incoming \texttt{receive()}, but that the cell state is still \texttt{EMPTY}. To resolve the race, \texttt{send(e)} updates the cell to \texttt{BUFFERED}, informing the receiver that the element is already available for retrieval (line~\ref{line:s:markBuffered}); thus, an elimination happens.

Finally, when the state is discovered in either \texttt{INTERRUPTED} or \texttt{BROKEN} state, \texttt{send(e)} cleans up the element field in the cell and restarts (lines~\ref{line:s:intOrBroken0}--\ref{line:s:intOrBroken1}).


\paragraph{Algorithm for \texttt{receive()}.}
The \texttt{receive()} implementation is symmetric and updates the cell state according to the life-cycle diagram in Figure~\ref{figure:rendezvous_cell}, via the \texttt{updCellRcv(..)} function. The essential differences from \texttt{updCellSend(..)} are that it (1) it marks the cell \texttt{BROKEN} when one is \texttt{EMPTY} but a rendezvous should happen and (2) checks for \texttt{BUFFERED} marker.

While this algorithm lends itself to additional minor optimizations, we preserved its structure so that it is consistent with the buffered channel variant, presented next. 

\subsection{The Buffered Channel}\label{subsec:bc}
Essentially, a \emph{buffered channel} is a rendezvous channel with~a fixed-size buffer for storing elements. 
A buffered channel of capacity $C$ allows senders to send $C$ elements before suspending, storing them in a buffer, similarly to bounded blocking queues. When the buffer is full, \texttt{send(e)} suspends. The \texttt{receive()} operation either suspends if the buffer is empty (if $C>0$), or retrieves the first element and resumes the first waiting sender, adding its element to the buffer.

\paragraph{High-level idea.}
To support the semantics above, we maintain an additional counter \texttt{B}, initialized at the buffer capacity \texttt{C}, which indicates the end of the ``imaginary buffer'' in the array of cells. The \texttt{send(e)} operation compares the reserved index \texttt{s} with \texttt{B} and buffers the element if \texttt{s\:<\:B}, suspending otherwise. The \texttt{receive()} operation manipulates the first cell of the buffer, either retrieving an element or installing its coroutine for suspension {---} in both cases, the buffer capacity decreases. To restore it, \texttt{receive()} increments the counter \texttt{B} (obtaining the index \texttt{b} of the first cell after the buffer) and resumes the sender suspended in the cell \texttt{A[b]} if needed.

\begin{wrapfigure}[6]{r}{0.25\linewidth}
    \centering
    \vspace{-1.2em}
    \includegraphics[width=\linewidth]{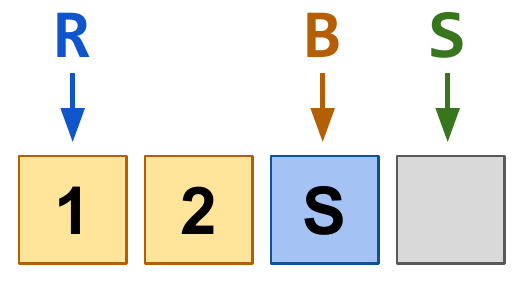}
\end{wrapfigure}

The picture on the right shows an example state for a buffered channel of capacity~$2$. Here, two elements are in the buffer, and one sender is waiting in the first cell after it (this cell is marked with <<\texttt{S}>>). The following \texttt{receive()} invocation increments \texttt{R} and retrieves the first element. After that, it expands the buffer by incrementing \texttt{B} and resuming the sender waiting in the corresponding cell \texttt{A[b]} (where \texttt{b} is the counter value right before the increment).

It is tempting to assume that the logical end of the buffer is always at position \texttt{R+C}, where \texttt{C} is the channel capacity, in which case there would be no need for the additional counter \texttt{B}. Unfortunately, this is incorrect due to interruptions. Consider a channel of capacity $1$. Two senders come: the first inserts its element to the buffer, while the second one suspends and interrupts.

\begin{wrapfigure}[6]{r}{0.25\linewidth}
    \centering
    \vspace{-1.15em}
    \includegraphics[width=\linewidth]{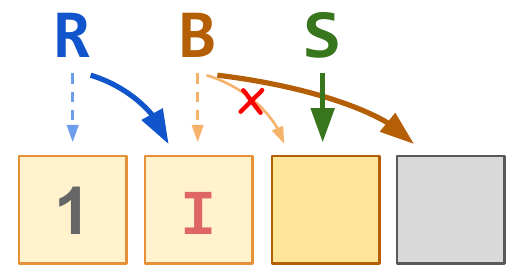}
\end{wrapfigure}

The resulting channel state is described in the right-hand-side embedded figure, where \texttt{R} and \texttt{B} reference cells by dashed arrows. After that, a \texttt{receive()} comes, incrementing \texttt{R} and retrieving the first element. If the new end of the buffer were \texttt{R+C}, the buffer (the cells between \texttt{R} and \texttt{R+C}) would cover the only cell with the interrupted sender. The following \texttt{send(e)} invocation would suspend, which is incorrect, as the channel is already empty.
Maintaining the end of the buffer explicitly solves the problem {---} if the cell \texttt{A[b]} stores an interrupted sender, the buffer expansion procedure restarts, incrementing the counter \texttt{B} again and, this way, fixing the issue. 

\vspace{-0.1em}
\paragraph{Buffer expansion.}
We extract the buffer expansion logic to a special \texttt{expandBuffer()} procedure, which is called after \texttt{receive()} successfully performs its synchronization, either retrieving the  first element, or storing its coroutine for suspension. When mentioning \texttt{receive()}, we usually refer only to its synchronization phase, without expansion.

Thus, \texttt{send(e)}, \texttt{receive()}, and \texttt{expandBuffer()} manipulate their own counters, and each cell can be processed by a single such call, apart for the case of interruption. 
To expand the buffer, we first increment the counter \texttt{B}, obtaining the index \texttt{b} right before the increment. In case the cell \texttt{A[b]} is not covered by a sender (\texttt{b\:$\geq$\:S}), it is guaranteed 
that the \texttt{send(e)} that will work with the cell later will observe that \texttt{s\:<\:B} and will buffer its element. Otherwise, if \texttt{b\:$<$\:S}, the cell \texttt{A[b]} stores a sender, or there is an incoming one.  

We now briefly discuss possible scenarios. 
Usually, the cell stores a suspended sender, and \texttt{expandBuffer()} tries to resume it, finishing on success. In case the sender resumption fails, the buffer expansion procedure restarts, as adding an already ``broken'' cell to the buffer does not help to expand it.
%
%
Also, the \texttt{send(e)} that works with the cell may still arrive, so the cell is empty. To inform the upcoming sender that it should not suspend, \texttt{expandBuffer()} moves the cell to a special \texttt{IN\_BUFFER} state and returns.
Another race occurs when \texttt{send(e)} had already observed that the cell is a part of the buffer and moved it to the \texttt{BUFFERED} state\footnote{This race is possible as \texttt{expandBuffer()} reads the \texttt{S} counter after the increment of \texttt{B}, so \texttt{send(e)} could already observe that the cell is already in the buffer; a similar race occurs between \texttt{send(e)} and \texttt{receive()}.}. 
As the cell is already in the buffer, the \texttt{expandBuffer()} procedure returns. 
Finally, it is possible for \texttt{receive()} to arrive earlier and process the cell. If there is a sender stored in the cell, \texttt{receive()} helps \texttt{expandBuffer()} to resume it. On success, both \texttt{receive()} and \texttt{expandBuffer()} finish, restarting if the sender is already interrupted. In case the receiver comes before the sender, it occupies the first buffer cell, while the \texttt{expandBuffer()} procedure can legitimately return.    

\begin{figure*}
    \centering
    \includegraphics[width=1\linewidth]{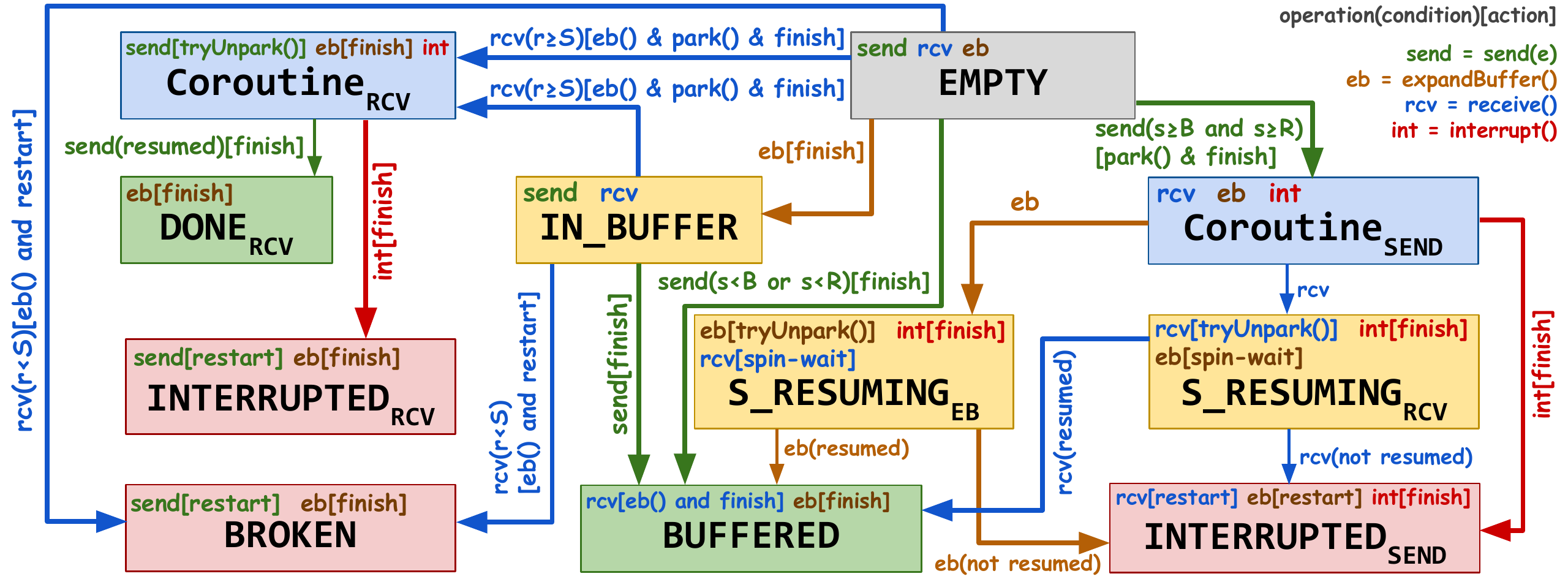}
    \vspace{-1.05em}
    \caption{
    \textbf{Cell life-cycle for the buffered channel.} The notation is the same as in the diagram for the rendezvous channel in Figure~\ref{figure:rendezvous_cell}. Here, we assume that it is possible to distinguish whether suspended sender (\texttt{Coroutine$_\mathtt{SEND}$}) or receiver (\texttt{Coroutine$_\mathtt{RCV}$}) is stored in the cell {---} this knowledge is critical for \texttt{expandBuffer()} to process the cell correctly. 
    }
    \vspace{-1em}
    \label{fig:buffered_cell_simplified}
\end{figure*}

\vspace{-0.1em}
\paragraph{Indistinguishable coroutines.}
For simplicity, we assume that it is possible to distinguish whether the coroutine stored in a cell is sender or receiver. 
While some languages, such as Go, provide this support, many
others, such as Kotlin or Java, require a more general implementation.
We discuss how to overcome this restriction in Appendix~\ref{appendix:nondistinguishable}.

\paragraph{Cell life-cycle.}
Figure~\ref{fig:buffered_cell_simplified} shows the cell life-cycle diagram for the buffered channels; the corresponding \texttt{updCellSend(..)} and \texttt{updCellRcv(..)} implementations are presented in Listing~\ref{listing:bc}. The abstract \texttt{send(e)} and \texttt{receive()} operations stay the same as for the rendezvous channels in Listing~\ref{listing:send_highlevel}.

\paragraph{Algorithm for \texttt{send(e)}.}
Initially, each cell is in the \texttt{EMPTY} state. When \texttt{send(e)} comes to an empty cell, it decides whether to buffer the element or suspend. 
In case the cell is part of the buffer (\texttt{s\:<\:B}), or a receiver is coming (\texttt{s\:<\:R}), it moves the cell state to \texttt{BUFFERED} and returns (lines~\ref{line:bc:inbuffer0}--\ref{line:bc:inbuffer1}). 
Similarly, \texttt{send(e)} buffers the element if the cell state~was moved to \texttt{IN\_BUFFER} by a concurrent \texttt{expandBuffer()} (line~\ref{line:bc:inbuffer0}).
Otherwise, the \texttt{send(e)} that works with the empty cell decides to suspend, moving its state to \texttt{Coroutine$_\mathtt{SEND}$} (lines~\ref{line:bc:suspend0}--\ref{line:bc:suspend1}).

In case the cell stores a suspended receiver, represented by the \texttt{Coroutine$_\mathtt{RCV}$} state, \texttt{send(e)} tries to make a rendezvous with it (lines~\ref{line:bc:receiver0}--\ref{line:bc:receiver1}). If the resumption succeeds, the operation finishes, restarting if the suspended receiver is already interrupted. Similarly, \texttt{send(e)} restarts if the cell state is \texttt{INTERRUPTED$_\mathtt{RCV}$} or \texttt{BROKEN}, skipping this cell (lines~\ref{line:bc:intbroken0}--\ref{line:bc:intbroken1}).

\begin{figureAsListingWide}
\vspace{-1em}
\begin{minipage}[t]{0.5\textwidth}
\begin{lstlisting}[basicstyle=\scriptsize\selectfont\ttfamily]
fun updCellSend(s: Int): Bool = while (true) { 
 state := A[s].state; r := R; @b := B@
 when {
 ^\indentrule^// Empty and either the cell is in the buffer 
 ^\indentrule^// or a receiver is coming ^\color{Mahogany}=>^ buffer
 ^\indentrule^state^\:^==^\:^null && (s^\:^<^\:^r^\:^||^\:^@s^\:^<^\:^b@) || @state^\:^==^\:^IN_BUFFER@: ^\label{line:bc:inbuffer0}^
 ^\indentrule^  if CAS(&A[s].state, state, BUFFERED):
 ^\indentrule^  ^\indentrule^  return true // successfully buffered ^\label{line:bc:inbuffer1}^
 ^\indentrule^// Empty, the cell is not in the buffer, 
 ^\indentrule^// no receiver is coming ^\color{Mahogany}=>^ suspend
 ^\indentrule^state == null && @s >= b@ && s >= r: ^\label{line:bc:suspend0}^
 ^\indentrule^  cor := curCor()
 ^\indentrule^  if CAS(&A[s].state, null, cor):
 ^\indentrule^  ^\indentrule^  cor.park({ onInterrupt(s) }); return true ^\label{line:bc:suspend1}^
 ^\indentrule^// Waiting receiver ^\color{Mahogany}=>^ try to resume it
 ^\indentrule^state is Coroutine^$_\mathtt{\textbf{RCV}}$^: ^\label{line:bc:receiver0}^
 ^\indentrule^  if state.tryUnpark():
 ^\indentrule^  ^\indentrule^  A[s].state = DONE^$_\mathtt{\textbf{RCV}}$^; return true
 ^\indentrule^  else: // interrupted, clean the cell and fail
 ^\indentrule^  ^\indentrule^  A[s].elem = null; return false ^\label{line:bc:receiver1}^
 ^\indentrule^//Interrupted receiver or poisoned ^\color{Mahogany}=>^ fail
 ^\indentrule^state == INTERRUPTED^$_\mathtt{\textbf{RCV}}$^ || state == BROKEN: ^\label{line:bc:intbroken0}^
 ^\indentrule^  A[s].elem = null // clean to avoid memory leaks
 ^\indentrule^  return false ^\label{line:bc:intbroken1}^
}
fun updCellRcv(r: Int): Bool = while (true) { 
 state := A[r].state; s := S
 when {
 ^\indentrule^// Empty and no sender is coming ^\color{Mahogany}=>^ suspend
 (state == null || @state == IN_BUFFER@) && r >= s: ^\label{line:bc:rcv:emptygood0}^
 ^\indentrule^  cor := curCor()
 ^\indentrule^  if CAS(&A[r].state, state, cor):
 ^\indentrule^  ^\indentrule^  ^^@expandBuffer()@  
 ^\indentrule^  ^\indentrule^  cor.park({ onInterrupt(r) }); return true ^\label{line:bc:rcv:emptygood1}^
 ^\indentrule^// Empty but a sender is coming ^\color{Mahogany}=> poison \& restart^
 ^\indentrule^(state == null || @state == IN_BUFFER@) && r < s: ^\label{line:bc:rcv:emptybad0}^
 ^\indentrule^  if ^^CAS(&A[r].state, state, BROKEN):
 ^\indentrule^  ^\indentrule^  ^^@expandBuffer()@; return false ^\label{line:bc:rcv:emptybad1}^
 ^\indentrule^// Buffered element ^\color{Mahogany}=>^ finish
 ^\indentrule^state == BUFFERED: @expandBuffer()@; return true ^\label{line:bc:rcv:buffered}^
 ^\indentrule^// Interrupted sender ^\color{Mahogany}=>^ fail
 ^\indentrule^state == INTERRUPTED^$_\mathtt{\textbf{SEND}}$^: return false ^\label{line:bc:rcv:intsender}^
 ^\indentrule^// Waiting sender ^\color{Mahogany}=>^ try to resume it
 ^\indentrule^state is Coroutine^$_\mathtt{\textbf{SEND}}$^: ^\label{line:bc:rcv:sender0}^
 ^\indentrule^  if @CAS(&A[r].state, state, S_RESUMING^$_{\mathtt{\textbf{RCV}}}$^)@: ^\label{line:bc:rcv:sresuming}^
 ^\indentrule^  ^\indentrule^  if state.tryUnpark(): // resumed
 ^\indentrule^  ^\indentrule^  ^\indentrule^  A[r].state = BUFFERED
 ^\indentrule^  ^\indentrule^  else: // interrupted
 ^\indentrule^  ^\indentrule^  ^\indentrule^  A[r].state = INTERRUPTED^$_\mathtt{\textbf{SEND}}$^ ^\label{line:bc:rcv:sender1}^
 ^\indentrule^// ^\color{Mahogany}expandBuffer()^ is resuming the sender ^\color{Mahogany}=>^ wait 
 ^\indentrule^#state == S_RESUMING^$_\mathtt{\textbf{EB}}$^: continue# ^\label{line:bc:rcv:spin}^
 }
}
\end{lstlisting}
\end{minipage}
\hfill
\begin{minipage}[t]{0.47\textwidth}
\begin{lstlisting}[firstnumber=54,basicstyle=\scriptsize\selectfont\ttfamily]
fun expandBuffer() = while(true) { 
 b := FAA(&B, +1) ^\label{line:bc:eb:inc}^ 
 if b >= S: return // not covered by ^\color{Mahogany}send()^, finish ^\label{line:bc:eb:check_send}^ 
 if updCellEB(b): return // update the cell state ^\label{line:bc:eb:updcelleb}^ 
}
// Returns `true` if ^\color{Mahogany}expandBuffer()^ should 
// finish, and `false` when it should restart
fun updateCellEB(b: Int): Bool = while(true) {
 state := A[b]
 when {
 ^\indentrule^// A suspended sender is stored ^\color{Mahogany}=>^ try to resume it
 ^\indentrule^state is Coroutine^$_\mathtt{SEND}$^: ^\label{line:bc:eb:sender0}^
 ^\indentrule^  if CAS(&A[s].state, state, S_RESUMING^$_\mathtt{\textbf{EB}}$^): ^\label{line:bc:eb:sresuming}^
 ^\indentrule^  ^\indentrule^  if state.tryUnpark():
 ^\indentrule^  ^\indentrule^  ^\indentrule^  A[s].state = BUFFERED; return true 
 ^\indentrule^  ^\indentrule^  else:
 ^\indentrule^  ^\indentrule^  ^\indentrule^  A[s].state = INTERRUPTED^$_\mathtt{\textbf{SEND}}$^; return false ^\label{line:bc:eb:sender1}^
 ^\indentrule^// The element is already buffered ^\color{Mahogany}=>^ finish
 ^\indentrule^state == BUFFERED: return true ^\label{line:bc:eb:buffered}^
 ^\indentrule^// The sender was interrupted ^\color{Mahogany}=>^ restart
 ^\indentrule^state == INTERRUPTED^$_\mathtt{\textbf{SEND}}$^: return false ^\label{line:bc:eb:interrupted}^
 ^\indentrule^// The cell is empty and ^\color{Mahogany}send()^ is coming ^\color{Mahogany}=>^
 ^\indentrule^// mark the cell as "in the buffer" and finish
 ^\indentrule^state == null: ^\label{line:bc:eb:empty0}^
 ^\indentrule^  if CAS(&A[b].state, null, IN_BUFFER): 
 ^\indentrule^  ^\indentrule^  return true ^\label{line:bc:eb:empty1}^
 ^\indentrule^// A receiver was stored in the cell ^\color{Mahogany}=>^ finish
 ^\indentrule^state is Coroutine^$_\mathtt{RCV}$^ || state == INTERRUPTED^$_\mathtt{\textbf{RCV}}$^: ^\label{line:bc:eb:rcv0}^
 ^\indentrule^  return true ^\label{line:bc:eb:rcv1}^
 ^\indentrule^// Poisoned cell ^\color{Mahogany}=> finish, receive() is in charge^
 ^\indentrule^state == BROKEN: return true ^\label{line:bc:eb:broken}^
 ^\indentrule^// A receiver is resuming the sender ^\color{Mahogany}=>^ wait
 ^\indentrule^#state == S_RESUMING^$_\mathtt{\textbf{RCV}}$^: continue# ^\label{line:bc:eb:spin}^
 }
}
// Cleans the cell on interruption
fun onInterrupt(i: Int) { ^\label{line:bc:onint0}^
 // Clean the element field
 A[i].elem = null
 // Update the state
 state := A[i].state
 when {
 ^\indentrule^// Suspended sender
 ^\indentrule^state is Coroutine^$_\mathtt{\textbf{SEND}}$^:
 ^\indentrule^  A[i].state = INTERRUPTED^$_\mathtt{\textbf{SEND}}$^
 ^\indentrule^// Suspended receiver
 ^\indentrule^state is Coroutine^$_\mathtt{\textbf{RCV}}$^:
 ^\indentrule^  A[i].state = INTERRUPTED^$_\mathtt{\textbf{RCV}}$^
 ^\indentrule^// Do nothing in any other case
 } ^\label{line:bc:onint1}^
}
\end{lstlisting}
\end{minipage}
\caption{
\textbf{The buffered channel algorithm.} Here we assume that \texttt{B} is initialized with the channel capacity. The abstract \texttt{send(e)} and \texttt{receive()} implementations stay the same as for the rendezvous channels in Listing~\ref{listing:send_highlevel} and, thus, are omitted. The key changes in them are highlighted with yellow. Also, the code related to the blocking behavior is highlighted with red. The \texttt{updCellSend(..)}, \texttt{updCellRcv(..)}, and \texttt{updCellEB(..)} functions as well as the interruption handler perform cell state transitions according to the diagram in Figure~\ref{fig:buffered_cell_simplified}. This code works under the assumption that it is possible to distinguish whether the stored coroutine is sender or~receiver; we propose a way to overcome this restriction in Appendix~\ref{appendix:nondistinguishable}.
}
\label{listing:bc}
\end{figureAsListingWide}

\paragraph{Algorithm for \texttt{receive()}.}
When \texttt{receive()} arrives at a cell in either \texttt{EMPTY} or \texttt{IN\_BUFFER} state, it decides to suspend, updating the cell state to \texttt{Coroutine$_\mathtt{RCV}$} and invoking the \texttt{expandBuffer()} procedure before the suspension (lines~\ref{line:bc:rcv:emptygood0}--\ref{line:bc:rcv:emptygood1}). Similar to rendezvous channels, the cell may be already covered by a concurrent \texttt{send(e)} (\texttt{r\:<\:S}), so suspension is forbidden. In this case, \texttt{receive()} moves the cell to the \texttt{BROKEN} state and restarts (lines~\ref{line:bc:rcv:emptybad0}\nobreakdash--\ref{line:bc:rcv:emptybad1}). As \texttt{receive()} poisons a buffer cell, it is crucial to call \texttt{expandBuffer()} after that (line~\ref{line:bc:rcv:emptybad1}). 

In case the cell stores a buffered element (\texttt{BUFFERED} state), \texttt{receive()} retrieves it and expands the buffer (line~\ref{line:bc:rcv:buffered}). 
If the cell stored a sender, and it was interrupted (\texttt{INTERRUPTED$_\mathtt{SEND}$} state), the operation restarts (line~\ref{line:bc:rcv:intsender}).
It is also possible for the receiver to find the cell in the \texttt{Coroutine$_\mathtt{SEND}$} state when coming earlier than \texttt{expandBuffer()} (lines~\ref{line:bc:rcv:sender0}--\ref{line:bc:rcv:sender1}). To make progress, \texttt{receive()} helps the late \texttt{expandBuffer()} and~tries to resume the sender. 
To synchronize with an \texttt{expandBuffer()}, it first moves the cell to an intermediate \texttt{S\_RESUMING$_\mathtt{RCV}$} state (line~\ref{line:bc:rcv:sresuming}), updating it to either \texttt{BUFFERED} (on success) or \texttt{INTERRUPTED$_\mathtt{SEND}$} (on failure); the further cell processing depends on this new state. 
Similarly, the \texttt{expandBuffer()} moves the cell to an intermediate state \texttt{S\_RESUMING$_\mathtt{EB}$}  when resuming a sender~{---} \texttt{receive()} waits in a spin-loop until the~cell state changes to either \texttt{BUFFERED} or \texttt{INTERRUPTED$_\mathtt{SEND}$} (line~\ref{line:bc:rcv:spin}). 


\paragraph{Algorithm for \texttt{expandBuffer()}.}
The procedure begins by incrementing its \texttt{B} counter (line~\ref{line:bc:eb:inc}). When no sender is waiting on the channel, the cell \texttt{A[b]} is usually not covered by \texttt{send(e)} (\texttt{b\:$\geq$\:S}), so no further actions are required and the buffer expansion finishes (line~\ref{line:bc:eb:check_send}). Otherwise, either the cell stores a sender or there is an upcoming one: the algorithm updates the cell according to the diagram in Figure~\ref{fig:buffered_cell_simplified} (line~\ref{line:bc:eb:updcelleb}).

In case the cell stores a suspended sender (\texttt{Coroutine$_\mathtt{SEND}$} state), \texttt{expandBuffer()} tries to resume it, moving the element to the buffer (lines~\ref{line:bc:eb:sender0}--\ref{line:bc:eb:sender1}).
To synchronize with a concurrent \texttt{receive()}, it first moves the cell to an intermediate \texttt{S\_RESUMING$_\mathtt{EB}$} state (line~\ref{line:bc:eb:sresuming}). On success, \texttt{expandBuffer()} updates the cell state to \texttt{BUFFERED} and finishes, moving the cell to \texttt{INTERRUPTED$_\mathtt{SEND}$} and restarting on failure. Similarly, when the cell state is already \texttt{BUFFERED}, the buffer expansion procedure finishes (line~\ref{line:bc:eb:buffered}), while observing the \texttt{INTERRUPTED$_\mathtt{SEND}$} state causes its restarting (line~\ref{line:bc:eb:interrupted}).

It is also possible for the cell to be \texttt{EMPTY} when the sender is still coming. 
To inform that the cell is a part of the buffer, \texttt{expandBuffer()} updates its state to \texttt{IN\_BUFFER} (lines~\ref{line:bc:eb:empty0}\nobreakdash--\ref{line:bc:eb:empty1}); \texttt{receive()} processes the \texttt{IN\_BUFFER} state as it were~\texttt{EMPTY}.
 
Another rare situation occurs when there was a receiver suspended in the cell. Independently on whether it is still suspended (\texttt{Coroutine$_\mathtt{RCV}$}), already resumed (\texttt{Done$_\mathtt{RCV}$}), or interrupted (\texttt{INTERRUPTED$_\mathtt{RCV}$}), the buffer expansion finishes (lines~\ref{line:bc:eb:rcv0}--\ref{line:bc:eb:rcv1}).
Similarly, the procedure finishes when the cell is in the \texttt{BROKEN} state (line~\ref{line:bc:eb:broken}); the receiver that poisoned the cell is in charge of adding an additional one to the buffer.

At last, when \texttt{receive()} resumes a sender, it first moves the cell state to \texttt{S\_RESUMING$_\mathtt{RCV}$} {---} \texttt{expandBuffer()} waits in a spin-loop until the cell state changes to either \texttt{BUFFERED} or \texttt{INTERRUPTED$_\mathtt{SEND}$} (line~\ref{line:bc:eb:spin}).

\paragraph{Interruptions.} When \texttt{send(e)} or \texttt{receive()} suspended~in the $i$-th cell is interrupted, the \texttt{onInterrupt(i)} function~is called (lines~\ref{line:bc:onint0}--\ref{line:bc:onint1}); it replaces the stored coroutine with \texttt{INTERRUPTED$_\mathtt{SEND}$} or \texttt{INTERRUPTED$_\mathtt{RCV}$} to avoid memory~leaks.

\subsection{Infinite Array Implementation} \label{section:infarr}
The last open question is how to emulate the infinite array. Notice that all cells are processed in sequential order, requiring access only to those between the counters. The high-level idea is to maintain a linked list of segments, each has an array of $K$ cells. All segments are marked with a unique \texttt{id} and follow each other;
Figure~\ref{figure:cell_storage} shows the high-level structure.
To access the $i$-th cell in the infinite array, we have to find the segment with the id \texttt{i\:/\:K} (adding it to the linked list if needed) and go to cell \texttt{i\:\%\:K} in it. 
Importantly, segments full of interrupted cells can be physically removed from the list in constant time {---} the algorithm maintains an additional \texttt{prev} pointer for that. Due to space constraints, we present the implementation details in Appendix~\ref{appendix:infarr}. 

\begin{figure}[H]
    \centering
    \includegraphics[width=0.5\linewidth]{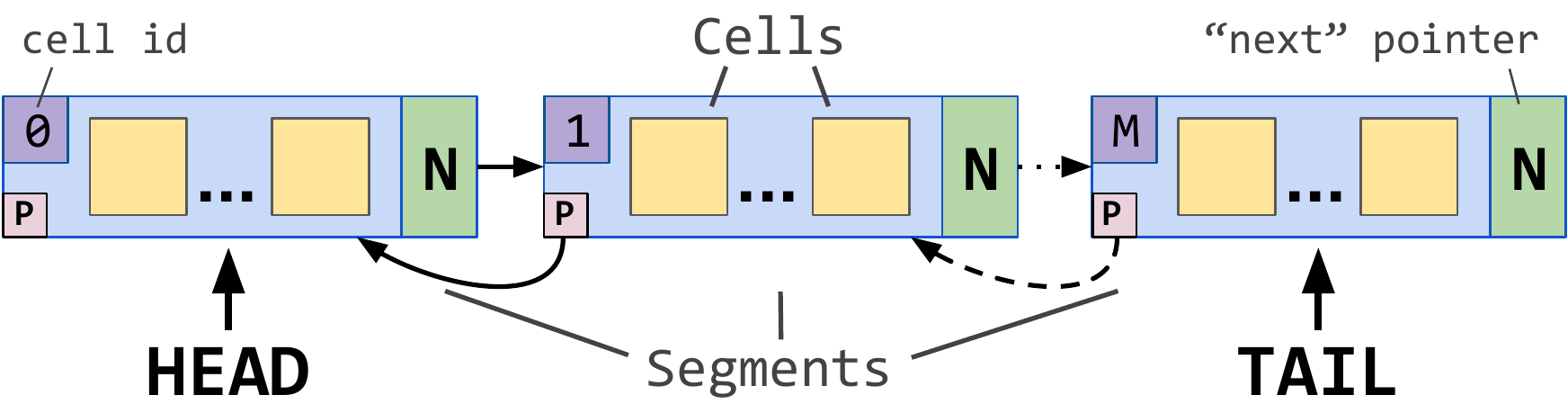}
    \captionof{figure}{An example of the infinite array structure.}
    \label{figure:cell_storage}
\end{figure}





\section{Correctness and Progress Guarantees}
Channels are blocking data structures by design (e.g., senders wait for receivers and vice-versa in rendezvous channels). To discuss correctness and progress guarantees, we apply the dual data structures formalism~\cite{scherer2004nonblocking}, initially introduced for synchronous queues.
Briefly, the idea is to split each operation into two phases: (1) the \emph{registration} phase, which either performs the operation or registers it as a waiter, and (2) the \emph{follow-up} phase, which happens after the operation is resumed; each phase has its own linearization point. This way, we split the blocking \texttt{send(e)} and \texttt{receive()} into two phases at the point of suspension.
As the \emph{follow\nobreakdash-up} parts are trivial (they simply complete the operations), we do not consider them in our discussion, focusing on the \emph{registration} phases that, essentially, perform all the synchronization.

\subsection{Correctness}\label{subsec:correctness}
Several works~\cite{LCRQ,SPDQ,YM16} use the idea of building a queue-like data structure on top of an infinite array with positioning counters for enqueuing and dequeuing. As these works already provide linearizability proofs, we reuse this knowledge, omitting the corresponding part. 
Specifically, it is known that both \texttt{enqueue(e)} and \texttt{dequeue()} linearize at the points where their counter increments, ignoring possible unsuccessful attempts to perform the operation. 
Similarly, \texttt{send(e)} linearizes on the \texttt{S} counter increment (line~\ref{line:s:inc} in Listing~\ref{listing:send_highlevel}) 
if the following \texttt{updCellSend(..)} invocation returns \texttt{true}, while \texttt{receive()}  linearizes on the  increment of \texttt{R} (line~\ref{line:rcvRend:inc_r}).
 To show that our rendezvous channel implementation is correct, we prove that neither \texttt{send(e)} nor \texttt{receive()} suspends when it should not {---} the rest of the algorithm works similar to a queue. After that, we discuss the buffer maintenance correctness in our buffered channel.  

\paragraph{Rendezvous channels.}
The \texttt{send(e)} operation is eligible to suspend only if \texttt{S\:$\geq$\:R} at the point the operation increments the \texttt{S} counter. As performing the increment and reading the opposite counter is non-atomic, it is enough to prove the claim below to show \texttt{send(e)} correctness. 

\begin{lemma}\label{theorem:rend}
If \texttt{send(e)} decides to suspend, it is guaranteed that \texttt{S\:$\geq$\:R} at the point where the \texttt{S} counter is incremented.
\end{lemma}
\begin{proof}
When \texttt{send(e)} increments its \texttt{S} counter (line~\ref{line:s:inc} in~Listing~\ref{listing:send_highlevel}), it obtains its value right before the increment, which we denote by \texttt{s}; let the value of \texttt{R} be \texttt{r'} at this time. 
Next, \texttt{send(e)} reads the  counter \texttt{R}, obtaining value \texttt{r\:$\geq$r'} (line~\ref{line:s:read_r}),
deciding to suspend if \texttt{s\:$\geq$\:r}. Then, as \texttt{r\:$\geq$r'}, the original suspension condition \texttt{s\:$\geq$\:r'} must also be satisfied. 
\end{proof}

\noindent
The discussion for \texttt{receive()} is symmetric, so we omit~it.

\paragraph{Suspensions in buffered channels.}
As in rendezvous channels, \texttt{receive()} is eligible to suspend only if \texttt{R\:$\geq$\:S} right before the increment of \texttt{R}, while \texttt{send(e)} can suspend only if \texttt{S\:$\geq$\:B} before the increment of \texttt{S}. The discussion of correctness is similar to  rendezvous channels.

Yet, dividing \texttt{receive()} into two phases might cause extra suspensions for concurrent \texttt{send}s. As \texttt{expandBuffer()} starts after the \texttt{receive()} synchronization completes, its effect, including the \texttt{B} counter increment, may be postponed. Nevertheless, we find this relaxation practical. Most programming languages already add extra suspensions due to internal scheduling mechanisms; e.g., Go can pause coroutines at any code location. Thus, even if both \texttt{receive()} phases occur atomically, extra suspensions are still possible.

\paragraph{Buffer maintenance.}
To argue the correctness of the buffer maintenance, we prove that the buffer size is constant over time. We begin with a simplified version of the algorithm, analyze it, and then show that the real implementation in this paper can be viewed as an equivalent, optimized version.

Figure~\ref{figure:proof_scheme} presents the cell life-cycle of the simplified buffered channel algorithm. 
From the high-level, it also manipulates the counters \texttt{S}, \texttt{R}, and \texttt{B}, invoking \texttt{expandBuffer()} after the \texttt{receive()} synchronization finishes; we present the pseudocode in Appendix~\ref{appendix:proofs_code}.
For simplicity, we also assume that the buffer capacity is greater than zero, receivers never interrupt, and \texttt{expandBuffer()} always moves \texttt{EMPTY} cells to the \texttt{IN\_BUFFER} state; we revisit these relaxations~later.

 \begin{figure}[h!]
    \centering
    \includegraphics[width=0.7\linewidth]{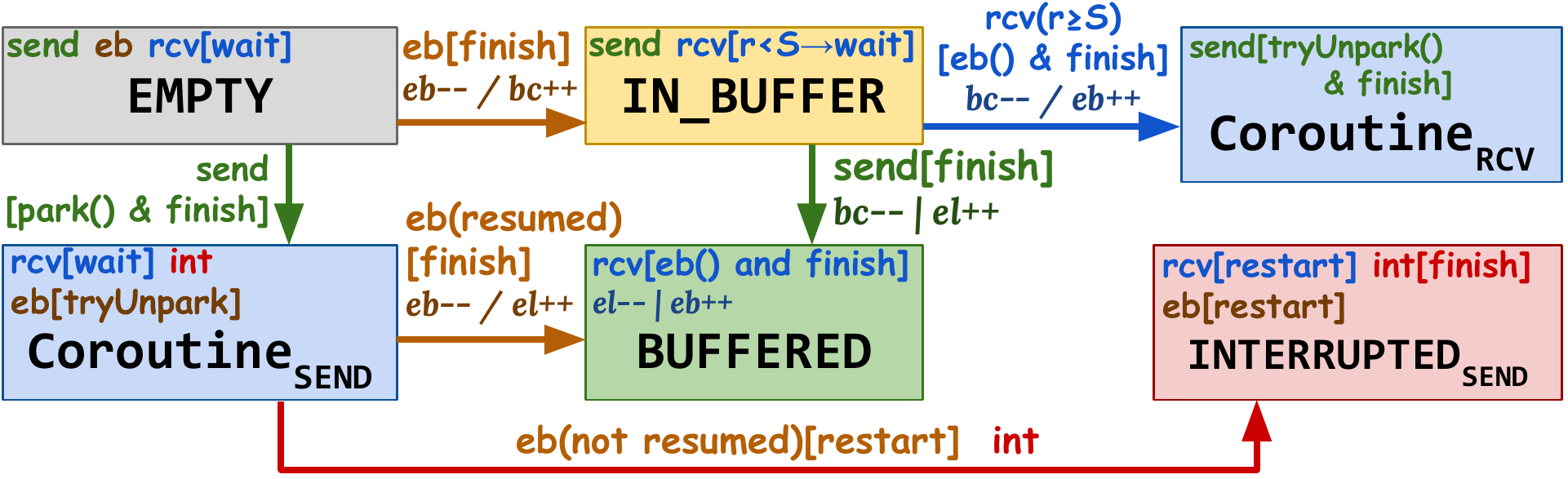}
    \captionof{figure}{Straightforward buffered channel scheme.}
    \label{figure:proof_scheme}
 \end{figure}

\noindent
To analyze the buffer maintenance, we count the number \texttt{bc} of empty buffer cells (\texttt{IN\_BUFFER} state), the number \texttt{el} of elements in the buffer (\texttt{BUFFERED} state), and the number \texttt{eb} of \texttt{expandBuffer()} calls that are yet to take effect. Given the initial buffer capacity \texttt{C}, we need to show that \texttt{bc\:+\:el\:+\:eb = C} at any algorithmic step.

\begin{theorem}\label{theorem:bc}
In the simplified algorithm, \texttt{bc\:+\:el\:+\:eb = C}.
\end{theorem}
\begin{proof}
Initially, the channel is empty and the first \texttt{C} cells are in the \texttt{IN\_BUFFER} state, so \texttt{bc\:=\:C} and \texttt{el\:=\:eb\:=\:0}. Now we discuss all possible cell state transitions shown in Figure~\ref{figure:proof_scheme}.

Usually, \texttt{send(e)} buffers its element by moving the cell state from \texttt{IN\_BUFFER} to \texttt{BUFFERED}, increasing the logical counter \texttt{el} and decreasing \texttt{bc}. 
When the cell stores a waiting receiver (\texttt{Coroutine$_\mathtt{RCV}$} state), \texttt{send(e)} resumes it and finishes; the counters do not change. 
Finally, if the cell state is \texttt{EMPTY}, \texttt{send(e)} suspends; the counters do not change~again.

The \texttt{receive()} operation always works with buffer cells, waiting in a spin-loop for update if the cell state is \texttt{EMPTY} or \texttt{Coroutine$_\mathtt{SEND}$}. 
When \texttt{receive()} retrieves a buffered element (\texttt{BUFFERED} state), it invokes \texttt{expandBuffer()}, so the \texttt{el} counter decreases and \texttt{eb} increases.
In case the cell is in the \texttt{INTERRUPTED$_\mathtt{SEND}$} state, the operation restarts; the counters stay the same.
Finally, \texttt{receive()} may suspend if the channel is empty (\texttt{r\:$\geq$\:S}), moving the cell state from \texttt{IN\_BUFFER} to \texttt{Coroutine$_\mathtt{RCV}$} and invoking \texttt{expandBuffer()} next; that results in decrementing \texttt{bc} and incrementing~\texttt{eb}.

As for \texttt{expandBuffer()}, it usually moves an \texttt{EMPTY} cell to the \texttt{IN\_BUFFER} state and finishes, increasing the \texttt{bc} counter and decreasing \texttt{eb}.
If the cell stores a suspender sender, it tries to resume it. On success, the cell state moves to \texttt{BUFFERED} and the buffer expansion finishes, leading to the \texttt{el} counter increment and \texttt{eb} decrement. Otherwise, if the resumption fails, or the cell is already in the \texttt{INTERRUPTED$_\mathtt{SEND}$} state, \texttt{expandBuffer()} restarts; this, the counters stay the same.

All the considered transitions keep the balance of the~counters, so bc\:+\:el\:+\:eb always equals the buffer capacity~\texttt{C}.
\end{proof}

Now we argue in detail that the buffered channel we presented (Listing~\ref{listing:bc}) only differs from this simplified version in terms of optimizations, and therefore Theorem~\ref{theorem:bc} above can be applied to it as well. 
First, it is clear from the diagram that allowing \texttt{receive()}-s to interrupt does not affect buffer maintenance {---} the \texttt{send(e)} that processes the cell simply restarts (line~\ref{line:bc:receiver1} and lines~\ref{line:bc:intbroken0}\nobreakdash--\ref{line:bc:intbroken1}). Second, if \texttt{expandBuffer()} does not move \texttt{EMPTY} cells to \texttt{IN\_BUFFER} state when the cell is not covered by a sender (the \texttt{b\:$\geq$\:S} condition at line~\ref{line:bc:eb:check_send}), \texttt{send(e)} will observe that the cell is in the buffer by the symmetric \texttt{s\:<\:B} check (line~\ref{line:bc:inbuffer0}). 

Now we cover the additional cell state transitions. 
When \texttt{receive()} suspends in the \texttt{EMPTY} cell without waiting until it becomes \texttt{IN\_BUFFER} (lines~\ref{line:bc:rcv:emptygood0}--\ref{line:bc:rcv:emptygood1}), \texttt{expandBuffer()} observes this optimization and finishes (lines~\ref{line:bc:eb:rcv0}--\ref{line:bc:eb:rcv1}). By poisoning the cell (lines~\ref{line:bc:rcv:emptybad1}--\ref{line:bc:rcv:emptybad1}, a new cell should be added to the buffer, so \texttt{receive()} correctly invokes \texttt{expandBuffer()} to maintain the buffer capacity (line~\ref{line:bc:rcv:emptybad1}).
Finally, \texttt{receive()} may resume a waiting sender (lines~\ref{line:bc:rcv:sender0}--\ref{line:bc:rcv:sender1}), which is an optimization to help the upcoming \texttt{expandBuffer()}.


\subsection{Progress Guarantees}
In the rendezvous channel, both \texttt{send(e)} and \texttt{receive()} are non-blocking. However, a \texttt{send-receive} pair can interfere infinitely often by poisoning cells over and over, so we can only formally guarantee obstruction freedom. Notwithstanding, our experiments suggest that cell poisoning is a very infrequent event in practice.
In the buffered channel algorithm, the \texttt{receive()} operation may be blocked and wait in a spin-loop until the cell state is updated by a concurrent \texttt{expandBuffer()}, and vice versa. 
As a result, the proposed rendezvous channel algorithm is obstruction-free, while the buffered channel is blocking. However, if the cell poisoning and the blocking behavior caused by a race between \texttt{receive()} and \texttt{expandBuffer()} were not taken into account as insignificant events, both algorithms would be wait-free. In case of interruptions, the guarantee would be reduced to lock-freedom due to segment removals.

\section{Evaluation}\label{section:experiments}
We implemented the buffered channel algorithm and integrated it into the standard Kotlin Coroutines library. Importantly, we extended the presented design to full channel semantics, including \texttt{trySend(e)} and \texttt{tryReceive()} operations, as well as adding functionality to close channels: after a channel is closed, \texttt{send()}s are forbidden. 
For performance comparison, we use the standard channel implementations in Kotlin Coroutines, as well as the algorithms by Scherer et\:al.\:\cite{scherer2006scalable} and Koval et\:al.\:\cite{koval2019channels} for the rendezvous channel analysis. 
We omit lock-based solutions, as Koval et al.\:\cite{koval2019channels} already compared their implementation against such designs, showing significantly better performance, and our approach  outperforms Koval et al.\:\cite{koval2019channels}. 

\paragraph{Environment.}
The experiments were run on a server with 4 Intel Xeon Gold 5218 (Cascade Lake) 16-core sockets,
for a total of 128 hardware threads with hyper-threading. 
%
In both our algorithm and the one by Koval et al. \cite{koval2019channels}, we have chosen the segment size of $32$, based on minimal tuning.



\paragraph{Benchmark.}
To evaluate performance, we use the classic producer-consumer workload: multiple coroutines share the same channel and apply a series of \texttt{send(e)} and \texttt{receive()} operations to it.
We use the same number of producer and consumer coroutines; the total number of coroutines either equals the number of threads or is fixed to \texttt{1000}.
As for the methodology, we measure the time it takes to transfer $10^6$ elements and count the system throughput. 
Also, we simulate some work between operations by consuming 100 non-contended loop cycles on average (following a geometric distribution), which decreases the contention on the channel. As for the buffer size, we chose $64$ as a standard size constant in many applications. Experiments with different buffer sizes show similar results, so we omit them.


\paragraph{Results.}
Figure~\ref{fig:prodcons} shows the benchmark results.
We first remark upon scalability: while other algorithms degrade throughput significantly as thread count increases, our approach continues to scale, and outperforms alternatives by up to \texttt{9.8}$\times$.
Interestingly, our buffered channel algorithm shows lower throughput than the rendezvous-only version, at higher thread counts. This is caused by the higher contention {---} in case of rendezvous channels, more threads are suspended, which makes synchronization cheaper.

\begin{figure}[h!]
    \centering
    \includegraphics[width=0.8\linewidth]{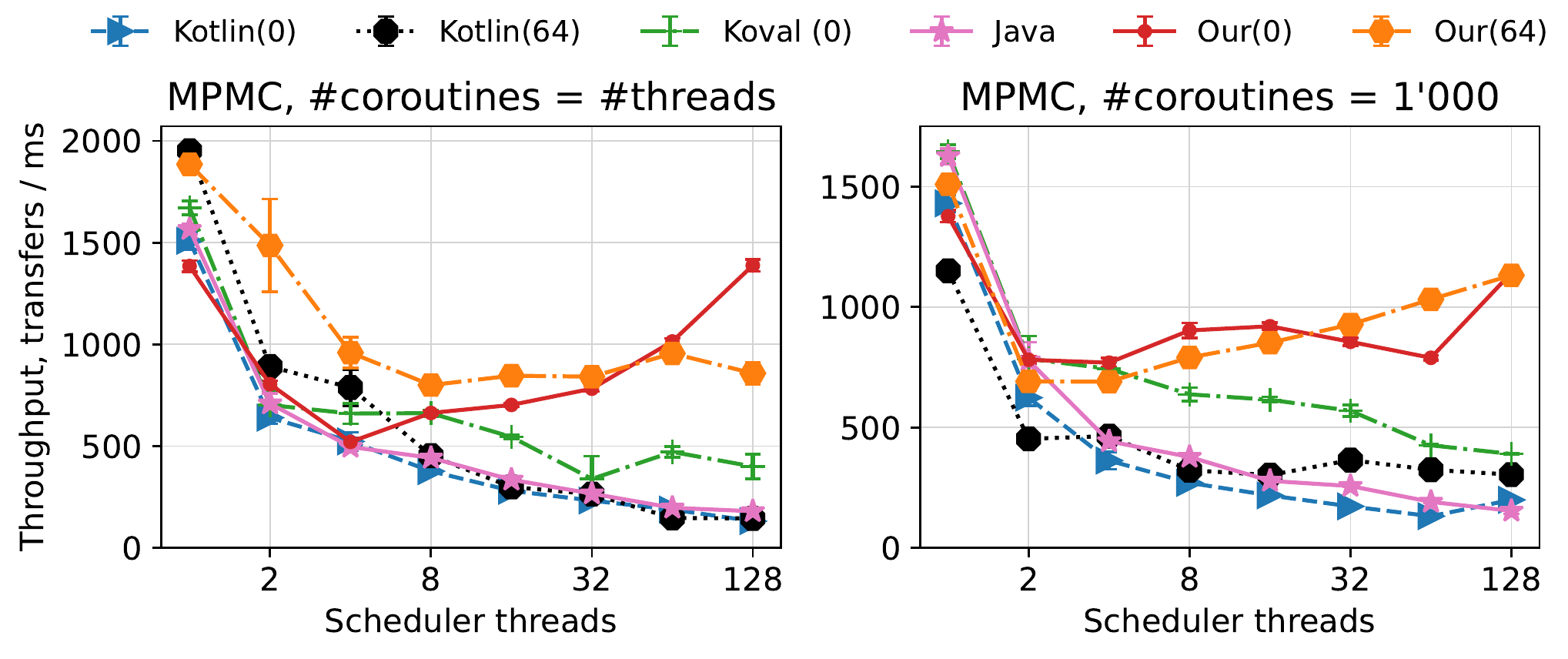}
    \captionof{figure}{
    Evaluation of channel algorithms on the producer-consumer benchmark. The buffer capacity is specified in brackets, \texttt{0} for rendezvous channels. \textbf{Higher is better.}
    }
    \label{fig:prodcons}
\end{figure}

\paragraph{Cell poisoning.} 
We collected statistics on the number of poisoned (\texttt{BROKEN}) cells. We observed that it never exceeds \texttt{10\%} of the total number of cells, even under extreme~contention. 

\paragraph{Memory usage.} We also collected the allocation pressure statistics. Under low contention, our rendezvous channel shows the same allocation rate as the algorithm by Koval et al. {---} they both leverage the linked list of segments. The Java synchronous queue and the standard Kotlin solutions show \texttt{40\%} and~\texttt{115\%} overhead, respectively. Under high contention, our solution is the best, while others show overhead from \texttt{20\%} (Java) to \texttt{90\%} (Kotlin).
As for buffered channels, the default Kotlin solution wins, as it reuses a pre-allocated array for the buffer.
In future work, we plan to reuse segments, further reducing memory overheads. 
\section{Related Work}\label{sec:related}
\paragraph{Programming languages and libraries.}
Most modern programming languages provide native support for asynchronous programming, either via actors~\cite{agha1985actors} or coroutines~\cite{kahn1976coroutines}. 
 Kotlin~\cite{kotlincoroutines}, Java~\cite{loom}, Go~\cite{golang}, and Rust~\cite{rust} use coroutines, while Akka~\cite{akka} and Erlang~\cite{erlang} use actors. All require channel implementations. Most solutions we are aware of use \emph{mutual exclusion} to maintain the waiting queue and the buffer. 
One notable exception is Kotlin, where the waiting queue~is built on top of a doubly-linked list of Sundell and Tsigas~\cite{sundell2004lock}. 
However, their buffered channel solution uses \emph{coarse-grained locking} to protect the buffer. The resulting implementation is exceptionally complex and shows significant overheads compared to our algorithm. 

\paragraph{Fair synchronous queues.}
Rendezvous channels are also known as \emph{fair synchronous queues}.
Hanson suggested the first solution based on three semaphores~\cite{hanson1996c},~which was later improved in Java 5 using a global lock with the \texttt{wait/notify} mechanism to make a rendezvous.
Java~6 implements a lock-free algorithm of  Scherer~et~al.~\cite{scherer2006scalable} based on the Michael-Scott queue~\cite{michael1996simple}. Koval et al. improved upon this by storing multiple waiters in each node and supporting interruptions more efficiently~\cite{koval2019channels}.

Izraelevitz and Scott~\cite{SPDQ} provided a general scheme for  non-blocking dual containers, presenting an efficient MPDQ synchronous queue algorithm; it leverages the LCRQ queue design~\cite{LCRQ}. 
While their approach is very interesting and efficient, some features prevent it from being applied in our setting:  both LCRQ~\cite{LCRQ} and MPDQ~\cite{SPDQ} require access to double-width atomic primitives, which are unavailable in most programming languages. 
Second, MPDQ breaks the channel semantics in terms of the conditions under which operations may suspend, leading to counter-intuitive executions. (See Appendix~\ref{appendix:mpdq} for examples.) 

Notably, all these synchronous queues solve only the rendezvous channel problem; extending them to the buffered channel semantics is non-trivial. Our work fills the gap. 

\paragraph{Unfair synchronous queues.}
Scherer et al. also proposed an \emph{unfair} stack-based synchronous queue in~\cite{scherer2006scalable}. This approach still induces a sequential bottleneck on the stack. To work around a single point of synchronization, Afek et al. introduced \textit{elimination-diffraction trees}~\cite{afek2010scalable}, in which internal nodes are \textit{balancer objects}~\cite{shavit2000combining} while leaves are synchronous queues.  
Hendler et al.~\cite{hendler2010scalable} approached the problem from a different angle by applying the \textit{flat combining} technique~\cite{hendler2010flat}.



\paragraph{Concurrent queues.}
Several techniques are known for implementing concurrent queues, and this is still an active research area, e.g.~\cite{michael1996simple, LCRQ, YM16}. 
Our work builds on some of these ideas. In particular, the general structure of the infinite array of cells is built on top of the  Michael-Scott queue~\cite{michael1996simple}, while the basic of the \texttt{send(e)/receive()} design is inspired by the LCRQ algorithm~\cite{LCRQ}. At the same time, our solution brings several new ideas, enabling not only the rendezvous semantics but also buffering and interruptions support. 


\section{Discussion}
This paper presents a fast and scalable algorithm for both rendezvous and buffered channels with interruptions support. By successfully integrating our solution into the standard Kotlin Coroutines library, we confirm that its design can express real-world requirements. Relative to previous implementations, our algorithm outruns them by up to \texttt{9.8}$\times$.
In the future, we aim to improve the channel performance even more by relaxing its strict FIFO semantics.



\clearpage
\bibliographystyle{plain}
\bibliography{references}

\appendix \clearpage
\onecolumn

\section{Non-Distinguishable Coroutines Support}\label{appendix:nondistinguishable}
In Section~\ref{section:algo}, we presented the buffered channel algorithm, which works under the assumption that it is possible to distinguish whether the coroutine stored in a cell is sender or receiver. While some languages, such as Go, provide this support, many others, such as Kotlin or Java, require a more general implementation. Now we discuss how to overcome this restriction.

The possibility to differentiate senders and receivers is critical for \texttt{expandBuffer()} to process the cell correctly: it finishes immediately if the cell stores a receiver, trying to resume the coroutine if it is a sender and finishing only if the resumption succeeds. The \texttt{send(e)} and \texttt{receive()} operations do not rely on this feature, as only the opposite type of request can suspend in the cell. The key idea to remove the restriction is to delegate the \texttt{expandBuffer()} completion when it is impossible to distinguish whether the stored coroutine is sender or receiver. 

Figure~\ref{figure:buffered_cell_extra} presents the update for the original cell life-cycle diagram in Figure~\ref{fig:buffered_cell_simplified}. When sender or receiver suspends, it moves the cell to the same \texttt{Coroutine} state.
When \texttt{expandBuffer()} processes a cell that stores a coroutine, and the cell is already covered by \texttt{receive()} (\texttt{b\:<\:R}), it is impossible to understand whether it is a sender or a receiver. To work around the race, we introduce a new \texttt{Coroutine+EB} state, which stores a coroutine with a special <<EB>> marker {---} the upcoming operations help \texttt{expandBuffer()} to finish its work.
Thus, when \texttt{b\:<\:R} and the cell stores a suspended coroutine, \texttt{expandBuffer()} moves the cell to the \texttt{Coroutine+EB} state and finishes.
In case \texttt{send(e)} observes the \texttt{Coroutine+EB} state (the cell stores a suspended receiver) it ignores the <<EB>> marker {---} indeed, the buffer expansion would finish when observing the \texttt{Coroutine$_\mathtt{RCV}$} state.
The \texttt{receive()} operation processes the \texttt{Coroutine+EB} state similarly as the \texttt{Coroutine} one, trying to resume the sender, with the only difference that it invokes the \texttt{expandBuffer()} procedure if the resumption fails; \texttt{expandBuffer()} would restart observing the \texttt{INTERRUPTED$_\mathtt{SEND}$} state.

A similar problem occurs with the interruption support. The \texttt{onInterrupt(i)} function also leverages the knowledge whether the interrupted request is sender or receiver, moving the cell state to \texttt{INTERRUPTED$_\mathtt{SEND}$} or  \texttt{INTERRUPTED$_\mathtt{RCV}$}, respectively. In fact, with a common \texttt{Coroutine} state, it can move the cell only to the common \texttt{INTERRUPTED} one. Yet, \texttt{expandBuffer()} proceeds \texttt{INTERRUPTED$_\mathtt{SEND}$} or  \texttt{INTERRUPTED$_\mathtt{RCV}$} differently, restarting in the first case and finishing in the second one. Similarly, we introduce a new \texttt{INTERRUPTED+EB} state, additionally to the \texttt{INTERRUPTED} one. At the same time, when receiver processes the cell, it moves its state to either \texttt{BUFFERED} or \texttt{INTERRUPTED$_\mathtt{SEND}$}; thus, the \texttt{expandBuffer()} procedure knows that it was sender stored in the cell.

Our implementation for Kotlin Coroutines includes this part of the algorithm and is well-documented, so we omit the corresponding pseudocode.

\begin{figure}[h!]
 \centering
 \includegraphics[width=0.6\linewidth]{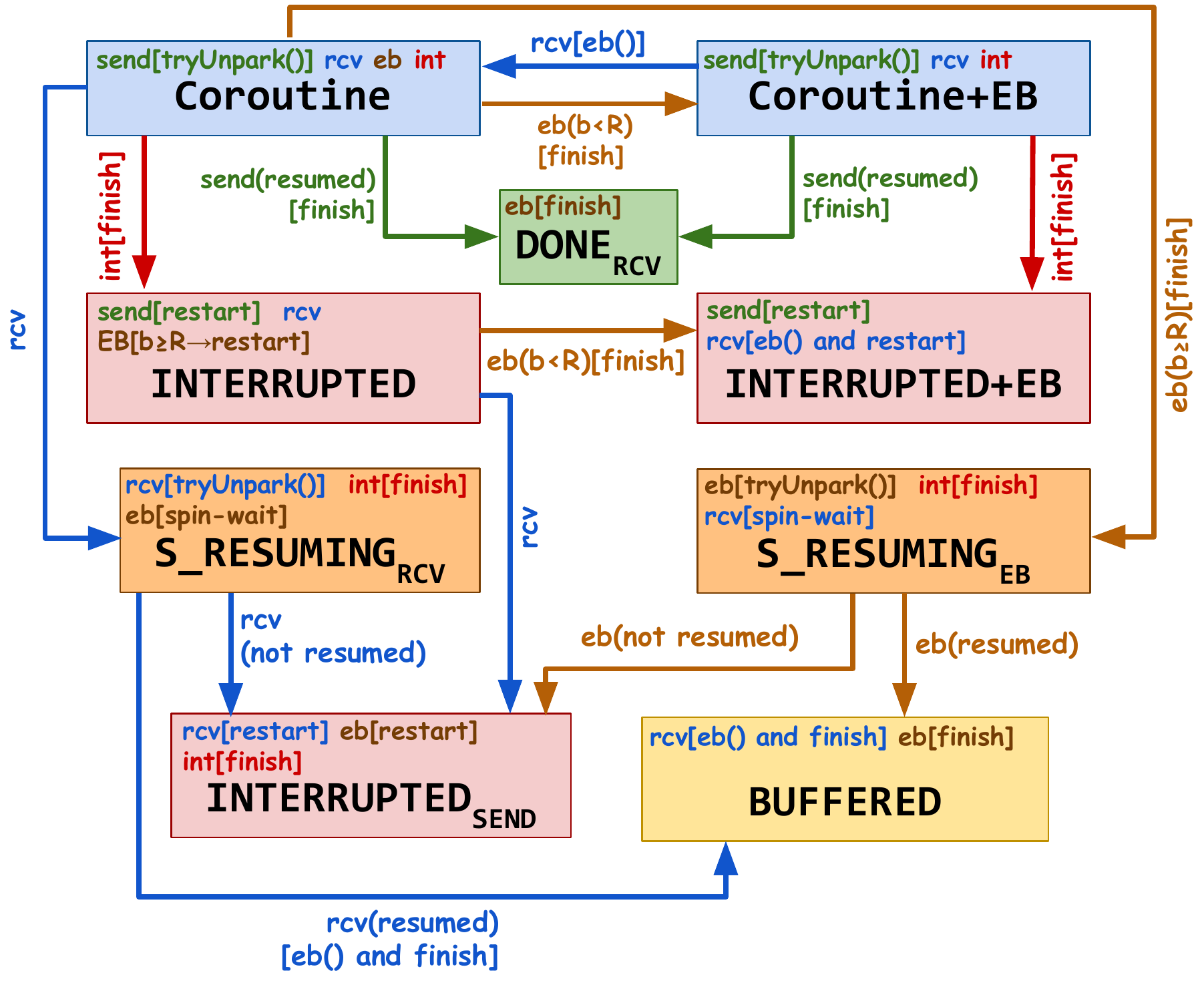}
 \captionof{figure}{Required updates to the cell life-cycle diagram for buffered channels in Figure~\ref{fig:buffered_cell_simplified} to overcome the restriction that it is possible to distinguish whether a coroutine stored in the cell is sender or receiver.}
 \label{figure:buffered_cell_extra}
\end{figure}
 


\clearpage
\section{Infinite Array: Implementation Details}\label{appendix:infarr}
In Subsection~\ref{section:infarr}, we briefly discussed how to build an infinite array with sequential access by each of the operation types. In short, we emulate it with a linked list of segments, each has an array of $K$ cells; Figure~\ref{figure:cell_storage2} below illustrates the data structure. All segments are marked with a unique \texttt{id}, maintaining an invariant that at the point of adding a new segment, the last one has id \texttt{TAIL.id\:+\:1}. To access cell $i$ in the infinite array, we have to find the segment with id \texttt{i\:/\:K} and go to the cell \texttt{i\:\%\:K} in it. 
Importantly, the segments full of interrupted cells can be physically removed from the list {---} the algorithm maintains a tricky \emph{doubly-linked} list, updating the \texttt{prev} pointers lazily and using them only for efficient removing. 

\begin{figure}[H]
    \centering
    \includegraphics[width=0.6\linewidth]{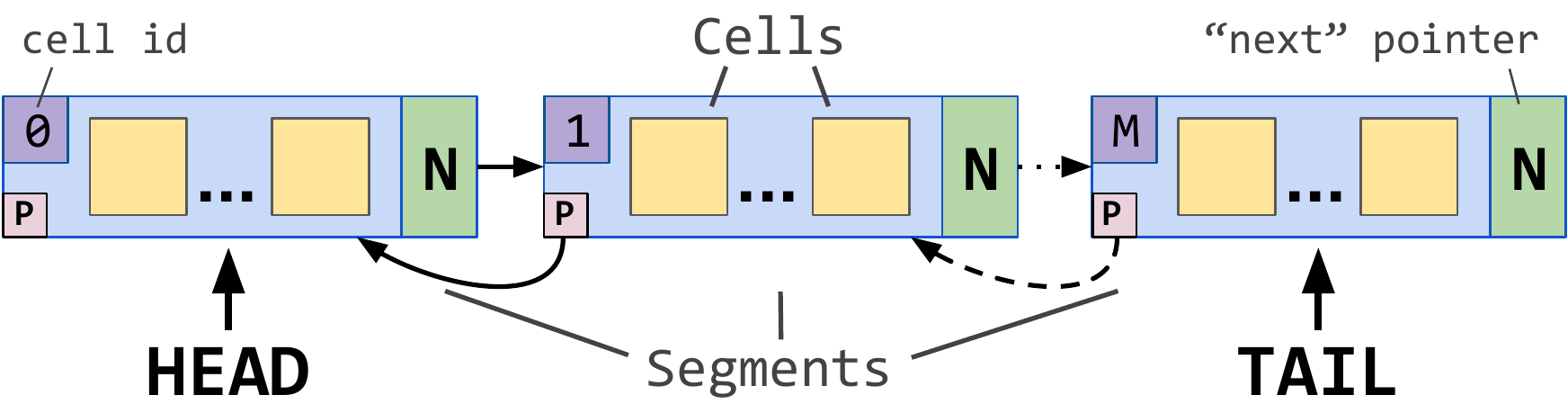}
    \captionof{figure}{An example of the infinite array structure.}
    \label{figure:cell_storage2}
\end{figure}

\paragraph{Required changes to \texttt{send(e)} and \texttt{receive()}.}
Few ingredients should be applied to \texttt{send(e)} and \texttt{receive()} to use the segments properly. 
First, to remove segments full of ``interrupted'' cells, a special \texttt{segm.onInterruptedCell()} function should be called upon interruption {---} roughly, it counts the number of interrupted coroutines stored in the segment, removing the latter if needed; we discuss the implementation later. 
Notably, as \texttt{expandBuffer()} needs to know whether the interrupted coroutine was sender or receiver, we should not call \texttt{segm.onInterruptedCell()} when the cell is already covered by \texttt{receive()} but is not processed by \texttt{expandBuffer()} yet, delegating the \texttt{segm.onInterruptedCell()} call to the upcoming \texttt{expandBuffer()}.

Second, when \texttt{send(e)} or \texttt{receive()} decides not to buffer the element or make a rendezvous, it should clean the \texttt{segm.prev} pointer of the working segment to ensure that processed segments are not reachable from the data structure anymore. 

Listing~\ref{listing:infarr_send} presents a pseudocode for \texttt{send(e)} in terms of these segments; the changes to \texttt{receive()} and \texttt{expandBuffer()} are symmetrical. 
Instead of maintaining head and tail pointers, as is usually done in concurrent queues~\cite{michael1996simple}, we maintain \texttt{SegmentS}, \texttt{SegmentR}, and \texttt{SegmentB} references to the segments last used by \texttt{send(e)}, \texttt{receive()}, and \texttt{expandBuffer()}. 

Initially, \texttt{send(e)} reads its last used segment  (line~\ref{line:sendSegm:readSegm1}) and increments \texttt{S} (line~\ref{line:sendSegm:inc_s}). Then, it locates the required segment by following the chain of \texttt{next} pointers, starting from the previously read one and updating \texttt{SegmentS}. The \texttt{findAndMoveForwardSend(..)} implementation is straightforward {---} it finds the required segment first, creating new segments if needed, and then updates \texttt{SegmentS} to the one that was found, skipping this part if it has already been updated to a later segment. When the required segment is already full of \texttt{INTERRUPTED} cells and is physically removed, the algorithm returns the first non-removed segment with \texttt{segm.id\:$\geq$\:id}. In the latter case (line~\ref{line:sendSegm:checkSegm}), the  ``reserved'' cell is already in \texttt{INTERRUPTED} state, so the caller should restart its operation. Instead of a simple restart, the algorithm can update the operation counter to the first non-interrupted cell \texttt{segm.id\:$\times$\:K} (line~\ref{line:sendSegm:skipCancelled}); thus, efficiently skipping the whole sequence of \texttt{INTERRUPTED} cells.

\begin{lstlisting}[label={listing:infarr_send},
caption={
Implementation of the \texttt{send(e)} operation, which manipulates a linked list of fixed-size segments. 
},
xleftmargin=1.7em
]
fun send(element: E) = while(true) {
 segm := SegmentS // read before the increment ^\label{line:sendSegm:readSegm1}^
 s := FAA(&S, +1)  ^\label{line:sendSegm:inc_s}^
 segm := findAndMoveForwardSend(segm, s / K) ^\label{line:sendSegm:readSegm2}^
 if segm.id != s / K: // the cell is INTERRUPTED ^\label{line:sendSegm:checkSegm}^
 ^\indentrule^ CAS(&S, s, segm.id * K) // skip INTERRUPTED cells^\label{line:sendSegm:skipCancelled}^
 ^\indentrule^ continue
 // The rest is same as before, but manipulates
 // the ^\color{Mahogany}`segm.cells[i]`^ cell instead, invokes  
 // ^\color{Mahogany}`segm.onInterruptedCell()`^  on interruption,  
 // and cleans ^\color{Mahogany}`segm.prev`^ pointer on rendezvous.
}
\end{lstlisting}

\paragraph{High-level implementation of the linked list of segments.}
Now we discuss how to manage these segments in the concurrent environment, especially when they can be removed due to interruptions to avoid memory leaks.
Listing~\ref{listing:remove_segment} presents a pseudo-code for managing the concurrent linked list of segments. 
On the left, a pseudo-code for \texttt{findAndMoveForwardSend(..)} is presented; the corresponding functions for \texttt{receive()} and \texttt{expandBuffer()} are almost identical and, therefore, are omitted. 
In brief, the operation finds the segment with the id equal or greater than the specified, and updates \texttt{SegmentS} if needed. However, this part becomes not so straightforward when coroutines can be interrupted. 
When the segment becomes full of \texttt{INTERRUPTED} cells, we want to physically remove it from the list to avoid memory leaks. This way, we can guarantee that the memory complexity depends only on the number of non-interrupted cells, even when all but one cell in each segment is \texttt{INTERRUPTED} since the segment size is constant. 


To remove a segment in $O(1)$ (under no contention), we maintain \texttt{prev} pointer to the \texttt{Segment} structure, which references to the first non-removed segment on the left, or equals \texttt{null} if all of them are removed or processed (e.g., when the removing segment is the head of the list). By maintaining the \texttt{prev} pointer, we can perform physical removal by linking the previous and the next segments to each other. However, doing so requires non-trivial tricks. 

Once a segment is removed, we should guarantee that it is no longer reachable by the \texttt{next} and \texttt{prev} references starting from \texttt{SegmentS}, \texttt{SegmentR}, and \texttt{SegmentB}. For this purpose, we split the removal procedure into two parts: \emph{logical} and \emph{physical}. 
We assume that the segment is logically removed if all the cells are in the \texttt{INTERRUPTED} state and neither \texttt{SegmentS} nor \texttt{SegmentR} or \texttt{SegmentB} references it (lines~\ref{line:r:removed0}--\ref{line:r:removed1}). At the same time, we need to guarantee that they cannot start referencing logically removed segments, making them ``alive'' again and causing memory leaks.

To solve this problem, we maintain the number of interrupted cells alongside the number of pointers that reference this segment in a single integer field (line~\ref{line:r:pc}). By storing these numbers in a single register, we are able to modify them atomically {---} to emphasize this, the corresponding code is wrapped with \texttt{atomic} block in the pseudocode. Thus, there are two ways for a segment to become logically removed. First, if neither \texttt{SegmentS} nor \texttt{SegmentR} or \texttt{SegmentB} references it, and the \texttt{interrupted} counter reaches \texttt{K}, the segment becomes logically removed and the following \texttt{remove()} call should remove it physically.
In the second case, all the cells are already in \texttt{INTERRUPTED} state, and the number of pointers that reference this segment reaches zero. In this case, the corresponding code must check whether the previously referenced segment became logically removed and invoke \texttt{remove()} if needed. 

Correspondingly, when \texttt{SegmentS} needs to be updated, they should increment the number of pointers that reference the new segment. However, if this new segment is already logically removed, the increment fails, and the update should be restarted. The corresponding logic is represented with the \texttt{tryIncPointers()} function (lines~\ref{line:r:tryInc0}--\ref{line:r:tryInc1}).
Similarly, when \texttt{SegmentS} stop referencing some segment, they decrement the number of pointers. The corresponding \texttt{decPointers()} function returns \texttt{true} if the segments becomes logically removed (lines~\ref{line:r:tryDec0}--\ref{line:r:tryDec1}).

The only exception from this is the attempt to remove the tail segment. We forbid removing the tail segment, as doing so would make it  more difficult to ensure that each segment has a unique id throughout the list's lifetime. Therefore, we ignore the attempts to remove the tail segment {---} see the first statement in \texttt{remove()} (line~\ref{line:r:checkTail}). Thus, if a segment was the tail of the list at the time of logical removal, the following \texttt{remove()} call does nothing, and the physical removal of this segment is postponed until it stops being the tail.

\paragraph{The \texttt{findAndMoveForwardSend(..)} operation.}
We split the operation into two parts. First, the \texttt{findSegment(..)} function finds the first non-removed segment with id equal to or greater than the requested one, creating new segments if needed (line~\ref{line:r:findSegm}). Once the segment is found, we try to make \texttt{SegmentS} point to it {---} this part can fail if the found segment becomes logically removed in the meantime, so the procedure restarts (line~\ref{line:r:move}).

Essentially, the \texttt{findSegment(..)} operation find the first non-removed segment with \texttt{id} equals or greater than the specified one following by \texttt{next} pointers, starting from the specified segment. In addition, once the tail of the list is updated, it checks whether the old tail should be removed (line~\ref{line:r:oldTailRemove}).

As for the \texttt{moveForwardSend(..)}, we need to increment and decrement the numbers of pointers. We first try to increment the number of pointers to the new segment (line~\ref{line:r:tryIncPointersMove}), returning \texttt{false} and causing \texttt{findAndMoveForwardSend(..)} to restart. If the increment of the number of pointers succeeds, the operation tries to update \texttt{SegmentS} to the new one (line~\ref{line:r:moveUpdate}). If the update succeeds, the number of pointers to the old segment (\texttt{cur} in the code) should be decremented, removing the segment physically if needed (line~\ref{line:r:removeOldTailMove}).
If the \texttt{SegmentS} update fails, the operation decrements the number of pointers to the new segment back (removing it if needed) and restarts.

\paragraph{The \texttt{Segment.onInterruptedCell()} operation.}
When the cell moves to \texttt{INTERRUPTED}, the \texttt{onInterruptedCell()} procedure is called. It increments the \texttt{interrupted} counter and checks if this led to the segment becoming logically removed, in which case it invokes \texttt{remove()} (lines~\ref{line:r:onc0}--\ref{line:r:onc1}).

\paragraph{The \texttt{Segment.remove()} operation.}
The final part is the \texttt{remove()} operation.
If the segment that is being removed is the tail, the removal is postponed and delegated to the \texttt{findSegment(..)} call that will update the tail and check whether the old one should be removed (line~\ref{line:r:checkTail}). 

Otherwise, the algorithm finds the first non-removed segment to the right (line~\ref{line:r:findR}) by following \texttt{next} pointers, and the first non-removed segment on the left (line~\ref{line:r:findL}) by following \texttt{prev} pointers. 
After that, we link the segment on the right with the segment on the left by updating its \texttt{prev} pointer (line~\ref{line:r:removeUpdPrev}). If a non-removed segment on the left has not been found, \texttt{prev} is updated to \texttt{null}. Otherwise, if such a segment has been found successfully, we link it with the segment on the right by updating the \texttt{next} pointer (line~\ref{line:r:removeUpdNext}).
Similar to the \texttt{findSegment(..)} operations, if all the segments on the right are logically removed, we manipulate the tail one.

As a result, we successfully linked the segments found on the left and on the right with each other. However, they could have been removed in the meantime. Therefore, we check that they are still non-removed and re-start the removal if not (lines~\ref{line:r:restart0}--\ref{line:r:restart1}). Otherwise, the removal procedure is completed. It is worth noting that it is possible that concurrent \texttt{remove()}-s keep the reference to our removed segment, and can accidentally re-link it with some other segment(s). However, due to checks that the segments found on the left and on the right are non-removed after the linking procedure (lines~\ref{line:r:restart0}--\ref{line:r:restart1}), we can guarantee that even if such an accident occurs, the \texttt{remove()} that led to this error will fix the problem. Thus, we know that the segment will be removed eventually.

\begin{figureAsListingWide}
\begin{minipage}[t]{0.52\textwidth}
\begin{lstlisting}[basicstyle=\scriptsize\selectfont\ttfamily]
fun findAndMoveForwardSend(start: Segment, 
                           id: Long): Segment {
  while (true): 
  ^\indentrule^  segm := findSegment(start, id) ^\label{line:r:findSegm}^
  ^\indentrule^  // Try to update `SegmentS`, and restart if  
  ^\indentrule^  // the found segment is removed and is not tail
  ^\indentrule^  if moveForwardSend(segm): return segm ^\label{line:r:move}^
}

fun findSegment(start: Segment, id: Long): Segment {
  cur := start
  while cur.id < id || ^^@cur.removed()@: ^\label{line:r:highlighted}^
  ^\indentrule^  if cur.next == null: 
  ^\indentrule^  ^\indentrule^  // Create a new segment if needed
  ^\indentrule^  ^\indentrule^  newSegm := Segment(id = cur.id + 1, prev = cur)
  ^\indentrule^  ^\indentrule^  if CAS(&cur.next, null, newSegm):
  ^\indentrule^  ^\indentrule^  ^\indentrule^  // Is the previous tail removed?
  ^\indentrule^  ^\indentrule^  ^\indentrule^  ^^@if cur.removed(): cur.remove()@ ^\label{line:r:oldTailRemove}^
  ^\indentrule^  cur = cur.next
  return cur
}
fun moveForwardSend(to: Segment): Bool {
  while (true):
  ^\indentrule^  cur := SegmentS // read the current segment
  ^\indentrule^  // Do we still need to update `SegmentS`?
  ^\indentrule^  if cur.id >= to.id: return true
  ^\indentrule^  // Try to inc pointers to `to`
  ^\indentrule^  if ^^@!to.tryIncPointers(): return false@ ^\label{line:r:tryIncPointersMove}^
  ^\indentrule^  // Try to update `SegmentS`
  ^\indentrule^  if CAS(&SegmentS, cur, to): ^\label{line:r:moveUpdate}^
  ^\indentrule^  ^\indentrule^  // Dec pointers to `cur` and remove if needed
  ^\indentrule^  ^\indentrule^  ^^@if cur.decPointers(): cur.remove()@ ^\label{line:r:removeOldTailMove}^
  ^\indentrule^  ^\indentrule^  return true
  ^\indentrule^  // The `SegmentS` update failed, dec pointers 
  ^\indentrule^  // for `to` back and remove it if needed
  ^\indentrule^  if ^^@to.decPointers(): to.remove()@
}

class Segment {
 // Initialized with (3, 0) for the first 
 // segment; stored in a single 32-bit Int
 var (pointers, interrupted) = (0, 0) ^\label{line:r:pc}^
 
 fun ^^^onInterruptedCell()^ { // invoked on interruption ^\label{line:r:onc0}^
   atomic { interrupted++ } 
   ^^@if removed(): remove()@
 } ^\label{line:r:onc1}^
\end{lstlisting}
\end{minipage}
\hfill
\begin{minipage}[t]{0.45\textwidth}
\begin{lstlisting}[firstnumber=48,basicstyle=\scriptsize\selectfont\ttfamily]
 fun removed(): Bool = atomic { ^\label{line:r:removed0}^
   return interrupted == K && pointers = 0
 } ^\label{line:r:removed1}^
 
 // Increments the number of pointers; fails
 // when the segment is logically removed
 fun tryIncPointers(): Bool = atomic { ^\label{line:r:tryInc0}^
   if removed():  return false
   pointers++; return true
 } ^\label{line:r:tryInc1}^
 // Decrements the number of pointers and
 // returns `true` if the segment becomes 
 // logically removed, `false` otherwise
 fun decPointers(): Bool = atomic { ^\label{line:r:tryDec0}^
   pointers--; return removed()
 } ^\label{line:r:tryDec1}^
  
 // Physically removes the current segment
 fun remove() = while(true) {
   // The tail segment cannot be removed
   if next == null: return ^\label{line:r:checkTail}^
   // Find the closest non-removed segments
   // on the left and on the right
   prev := aliveSegmRight() ^\label{line:r:findR}^
   next := aliveSegmLeft() ^\label{line:r:findL}^
   // Link next and prev
   next.prev = prev ^\label{line:r:removeUpdPrev}^
   if prev != null: prev.next = next ^\label{line:r:removeUpdNext}^
   // Are `prev` and `next` still non-removed?
   if next.removed() && next.next != null:  ^\label{line:r:restart0}^
   ^\indentrule^  continue
   if prev != null && prev.removed(): continue ^\label{line:r:restart1}^
   return // this segment is removed
 }
 fun aliveSegmLeft(): Segment? {
   cur := prev
   while cur != null && cur.removed():
   ^\indentrule^ cur = cur.prev
   return cur // `null` if all are removed
 }
 fun aliveSegmRight(): Segment {
   cur := next
   while cur.removed() && cur.next != null:
   ^\indentrule^ cur = cur.next
   return cur // tail if all are removed
 }
} // end of class Segment
\end{lstlisting}
\end{minipage}
\vspace{-1em}
\caption{Pseudocode for the segment management algorithm. First, the \texttt{findAndMoveForwardSend(..)} implementation is presented; the functions for \texttt{receive()} and \texttt{expandBuffer()} are similar and, thus, omitted. After that, the segment structure along with the \texttt{remove()} logic are showed. The key parts for the segment removal support are highlighted with yellow.}
\vspace{-3em}
\label{listing:remove_segment}
\end{figureAsListingWide}

\clearpage
\section{The Simplified Buffered Channel Algorithm for the Correctness Proof in Subection~\ref{subsec:correctness}} \label{appendix:proofs_code}
To show the buffer maintenance correctness in Subection~\ref{subsec:correctness}, we analyze a simplified version of the algorithm presented in the paper. In this appendix section, we show the pseudocode of this simplified algorithm. 

\begin{figureAsListingWide}
\begin{minipage}[t]{0.5\textwidth}
\begin{lstlisting}[basicstyle=\scriptsize\selectfont\ttfamily]
var S, R, B: Long
var A: InfiniteArray

fun send(element: E) = while(true) { 
 s := FAA(&S, +1) // reserve a cell 
 A[s].elem = element 
 if updCellSend(s): return
 // Restart the operation
}

fun receive(): E = while(true) {
 r := FAA(&R, +1) // reserve a cell 
 if updCellRcv(r): // success, get the element 
 ^\indentrule^  e := A[r].elem; A[r].elem = null; return e
 // Restart the operation
}

// Called when ^\color{Mahogany}`receive()`^ synchronization completes
fun expandBuffer() = while(true) {
  b := FAA(&B, +1)
  if updCellEB(b): return  
}

// Returns `false` if this ^\color{Mahogany}`send(e)`^ should restart
fun updCellSend(s: Int): Bool = while(true) { 
 state := A[s].state // read the current state 
 b := B // read the logical end of the buffer 
 when {
 ^\indentrule^// The cell is a part of the buffer ^\color{Mahogany}=>^ finish
 ^\indentrule^state == IN_BUFFER:
 ^\indentrule^  if CAS(&A[s].state, state, BUFFERED):
 ^\indentrule^  ^\indentrule^  return true  
 ^\indentrule^// Empty and won't be in the buffer ^\color{Mahogany}=>^ suspend
 ^\indentrule^state == null && s >= b: 
 ^\indentrule^  cor := curCor() 
 ^\indentrule^  if CAS(&A[s].state, null, cor): 
 ^\indentrule^  ^\indentrule^  cor.park( // wait for a rendezvous
 ^\indentrule^  ^\indentrule^   onInterrupt = {A[s] = (INTERRUPTED^$_\mathtt{SEND}$^, null)} 
 ^\indentrule^  ^\indentrule^  ) 
 ^\indentrule^  ^\indentrule^  return true 
 ^\indentrule^// Waiting receiver ^\color{Mahogany}=>^ resume it and finish
 ^\indentrule^state is Coroutine^$_\mathtt{RCV}$^: 
 ^\indentrule^  state.tryUnpark(); return true
 ^\indentrule^// Empty but will be in the buffer ^\color{Mahogany}=>^ spin-wait
 ^\indentrule^state == null && s < b: continue
 }
}
\end{lstlisting}
\end{minipage}
\hfill
\begin{minipage}[t]{0.46\textwidth}
\begin{lstlisting}[firstnumber=48,basicstyle=\scriptsize\selectfont\ttfamily]
// Returns `false` if this ^\color{Mahogany}`receive()`^ should restart
fun updCellRcv(r: Int): Bool = while(true) {
 state := A[r].state // read the current state
 s := S // read the sender's counter
 when {
 ^\indentrule^// Buffer cell and no sender is coming ^\color{Mahogany}=>^ suspend
 ^\indentrule^state == IN_BUFFER && r >= s: 
 ^\indentrule^  cor := curCor()
 ^\indentrule^  if CAS(&A[r].state, null, cor):
 ^\indentrule^  ^\indentrule^  expandBuffer()
 ^\indentrule^  ^\indentrule^  cor.park()
 ^\indentrule^  ^\indentrule^  return true
 ^\indentrule^// Buffer cell but a sender is coming ^\color{Mahogany}=>^ spin-wait
 ^\indentrule^state == IN_BUFFER && r < s: continue
 ^\indentrule^// Buffered element ^\color{Mahogany}=>^ finish
 ^\indentrule^^^state == BUFFERED: 
 ^\indentrule^  expandBuffer(); return true
 ^\indentrule^// Interrupted sender ^\color{Mahogany}=>^ fail
 ^\indentrule^state == INTERRUPTED^$_\mathtt{SEND}$^: return false
 ^\indentrule^// Suspended sender ^\color{Mahogany}=>^ spin-wait
 ^\indentrule^state is Coroutine^$_\mathtt{SEND}$^: continue
 ^\indentrule^// Empty cell ^\color{Mahogany}=>^ spin-wait
 ^\indentrule^state == null: continue
 }
}

// Returns `true` if ^\color{Mahogany}expandBuffer()^ should 
// finish, and `false` when it should restart
fun updCellEB(b: Int): Bool = while(true) {
 state := A[b]
 when {
 ^\indentrule^// Cell is empty ^\color{Mahogany}=>^ move to the buffer ^\color{Mahogany}\&^ finish
 ^\indentrule^state == null:
 ^\indentrule^  if CAS(&A[b].state, null, IN_BUFFER): 
 ^\indentrule^  ^\indentrule^  return true
 ^\indentrule^// Waiting sender  ^\color{Mahogany}=>^ try to resume it 
 ^\indentrule^state is Coroutine^$_\mathtt{SEND}$^:
 ^\indentrule^  if state.tryUnpark():
 ^\indentrule^  ^\indentrule^  A[s].state = BUFFERED; return true
 ^\indentrule^  else:
 ^\indentrule^  ^\indentrule^  A[s].state = INTERRUPTED^$_\mathtt{\textbf{SEND}}$^; return false
 ^\indentrule^// The sender was interrupted ^\color{Mahogany}=>^ fail
 ^\indentrule^state == INTERRUPTED^$_\mathtt{\textbf{SEND}}$^: return false
}
\end{lstlisting}
\end{minipage}
\caption{The simplified version of the buffered channel algorithm according to the cell life-cycle diagram in Figure~\ref{figure:proof_scheme}.
}
\label{listing:send_highlevel2}
\end{figureAsListingWide}

\clearpage
\section{Incorrect Execution when Using the MPDQ Synchronous Queue as a Rendezvous Channel}\label{appendix:mpdq}
In the MPDQ~\cite{SPDQ} synchronous queue algorithm, each operation reserves a cell in a circular array, which is then employed for rendezvous. 
Similar to our solution, the cells are reserved by incrementing a counter separate for senders and receivers. 
After the operation increments its counter, it either adds itself to the circular array and suspends or makes a rendezvous with the opposite request already stored in the cell. 
The key difference with our solution is that the MPDQ algorithm does not compare the operation counters \texttt{S} and \texttt{R} to decide whether to suspend or not, always suspending if the cell state is \texttt{EMPTY}.

Now we show that the MPDQ synchronization scheme introduces incorrect behavior. 
We consider an execution where three threads $s_1, s_2$ and $r_1$ execute \texttt{send(e)}, \texttt{send(e)}, and \texttt{receive()} operations, respectively. 
Assume that $s_1$ reserves a cell in the waiting queue by incrementing the corresponding counter but not adding itself to the cell yet. 
Then, $s_2$ also reserves the (following) cell in the queue, and \emph{does} store itself to the cell, completing the registration part of its \texttt{send(e)} operation and suspending. 
Finally, $r_1$'s \texttt{receive()} operation starts, reserves a ``dequeue'' cell (the one corresponding to $s_1$) to rendezvous with, 
and then suspends, waiting for the $s_1$'s operation to complete. However, this is incorrect relative to the channel semantics since there is another \texttt{send(e)} operation ($s_2$'s), which has already completed its registration phase and is suspended, waiting for a rendezvous. 

In particular, with coroutines, $s_1$ may be scheduled on one native thread while both $s_2$ and $r_1$ are scheduled on another. Thus, when the \texttt{send(e)} operation in $s_2$ suspends, the execution goes to $r_1$, which knows that there should an element in the channel {---} suspending in this case becomes extremely counter-intuitive. 

Interestingly, this type of pathology is exactly why ``broken'' cells are introduced in our algorithm. More precisely, in our algorithm, a \texttt{receive()} operation, which tries to make a rendezvous, marks the reserved but empty cell as broken and proceeds to rendezvous on the first populated cell.

\end{document}